\documentclass[10pt,a4paper]{article}

\usepackage[utf8]{inputenc}
\usepackage[T1]{fontenc}
\usepackage{amsmath,amssymb,amsfonts}
\usepackage{booktabs}
\usepackage{hyperref}
\usepackage{geometry}
\usepackage{enumitem}
\usepackage{xcolor}
\usepackage{graphicx}
\newcommand{\rev}[1]{{#1}}

\newcommand{\revthree}[1]{{#1}}

\newcommand{\revfive}[1]{{#1}}
\newcommand{\revsix}[1]{{#1}}
\newcommand{\revseven}[1]{{#1}}
\newcommand{\reveight}[1]{{#1}}

\newcommand{\revten}[1]{{#1}}

\usepackage{makecell}
\usepackage{multirow}
\usepackage{longtable}
\usepackage{placeins}
\usepackage{tabularx}
\usepackage{float}
\usepackage{amsthm}
\newtheorem{proposition}{Proposition}
\theoremstyle{exam}
\newtheorem{exam}{Example}
\theoremstyle{remark}

\theoremstyle{plain}
\usepackage{algorithm}
\usepackage{algpseudocode}
\algrenewcommand\algorithmicrequire{\textbf{Input:}}
\algrenewcommand\algorithmicensure{\textbf{Output:}}
\title{\textbf{RIDGE: An Autonomous Framework for Validation and Method Discovery in LLM-Generated Option Pricing}}
\author{Liexin Cheng$^{1,2}$ \quad Xue Cheng$^{2}$ \quad Shuaiqiang Liu$^{3}$\thanks{Email: \href{mailto:shuaiqiangliu@tudelft.nl}{shuaiqiangliu@tudelft.nl}} \quad Cornelis W.\ Oosterlee$^{1}$ \\[4pt]
	{\small $^{1}$Mathematical Institute, Utrecht University \quad $^{2}$School of Mathematical Sciences, Peking University}\\
	{\small $^{3}$Delft Institute of Applied Mathematics, Delft University of Technology}}
\date{}

\begin{document}
	\maketitle
	
\begin{abstract}
Automated code generation is becoming an important tool in quantitative finance, where large language models can generate option pricing implementations directly from mathematical model specifications. Validating such implementations, however, requires considerably more than conventional software testing: numerical pricing methods must remain mathematically consistent, numerically stable, and reliable across a wide range of model parameters.

We introduce RIDGE, an autonomous validation framework in which generated pricing implementations are subjected to structured no-arbitrage tests, stress tests, benchmark comparisons, and consistency checks. Validation evidence is interpreted diagnostically, while the resulting knowledge is accumulated in a repository and reused across models and successive validation iterations. This enables systematic refinement of both the pricing implementation and the validation methodology.

The framework is applied to five stochastic volatility models. Across these studies, all detected implementation defects are removed and, in two cases, the validation process reveals methodological limitations and motivates the development of alternative numerical methods. The supplementary material is available in the GitHub repository: \url{https://github.com/ShQiangLiu/ridge}.
\end{abstract}
	
	\section{Introduction}
	
Artificial intelligence has recently emerged as a useful tool for exploring program and algorithmic search spaces. AlphaCode~\cite{Li2022AlphaCode}, for example, demonstrated competition-level program synthesis through large-scale generation and filtering, while AlphaTensor~\cite{Fawzi2022AlphaTensor} and AlphaDev~\cite{Mankowitz2023AlphaDev} showed that learning-based agents can discover improved algorithms; see also~\cite{LeGoues2012GenProg,Breck2017MLTestScore,Udrescu2020AIFeynman}. These developments demonstrate the power of AI-driven generation under executable objectives. 
However, they remain largely empirical: they demonstrate that useful programs can be generated or discovered, but they do not explain when the generation process can be stably guided toward  valid  scientific-computing implementations.

This gap becomes particularly important for software generated by large language models (LLMs), which are increasingly used to generate scientific software directly from mathematical specifications~\cite{Lu2026,seo2025papercode,Khattab2024DSPy}.  In quantitative finance, they can generate implementations of option pricing methodologies with limited manual coding effort and can efficiently explore alternative modeling assumptions and numerical techniques.   This development raises a fundamental question: how can one determine whether a generated implementation is mathematically correct, numerically stable, and reliable beyond a small collection of benchmark examples?
	Standard software validation techniques~\cite{OberkampfRoy2010} provide limited assurance in this setting. A pricing implementation may produce acceptable values for standard test cases, yet fail in more demanding parameter regimes. Such failures may arise from unsuitable numerical parameters or from implementation errors. In some cases, important adaptations may simply be omitted. Furthermore, numerical pricing methods inevitably introduce approximation errors~\cite{Higham2002} and require truncation and discretization steps, whose combined effect is difficult to assess through unit tests alone. Validation therefore requires more than verifying that a small collection of numerical examples produces seemingly reasonable outputs.

	To address this problem, we introduce RIDGE ({R}epository-driven {I}terative {D}iagnosis, {G}eneration and {E}valuation), an autonomous validation framework for generated pricing implementations. RIDGE combines deterministic numerical measurements with repository-driven diagnosis and iterative implementation refinement. The framework combines consistency checks, stress testing, and benchmark comparisons, while the diagnostic knowledge accumulated during validation is retained and reused across validation iterations and model classes. As a consequence, both the pricing implementations and the validation methodology evolve as additional models are examined.
	The framework is developed and tested in the setting of stochastic-volatility option pricing. These models span a broad range of numerical complexity and  provide a natural testbed for studying autonomous validation. We consider five representative stochastic-volatility models: Heston~\cite{heston1993}, Bates~\cite{bates1996}, rough Heston~\cite{el2019rough}, Heston--Hull--White~\cite{grzelak2011,HaentjensInTHout2012}, and Generalized Stochastic Volatility Jump-Diffusion (G-SVJD)~\cite{fusari2025}.
	Option pricing is a central computational task in quantitative finance and serves as a cornerstone of modern derivatives practice. Some stochastic-volatility models admit highly efficient transform-based valuation methods, whereas others require approximations or Monte Carlo simulation. This diversity of numerical methodologies makes stochastic-volatility option pricing a natural testbed for autonomous validation.

	An outcome of the present study is that validation can, in certain cases, lead to alternative numerical techniques. In two cases, the Heston--Hull--White and G-SVJD models, the validation process indicates structural limitations of the initially generated Monte Carlo methodologies and leads instead to alternative semi-analytic valuation approaches based on conditional Gaussian representations and the Feynman--Kac theorem. The framework therefore serves both to validate generated implementations and to identify and develop better-suited numerical methodologies.

RIDGE formalizes domain-specific validation as the central mechanism for the reliable development of AI-generated option-pricing implementations. It combines iterative code generation and refinement with deterministic validation operators based on financial admissibility, benchmark hierarchies, degeneration limits, stress testing, and repository-driven knowledge transfer. While generic coding agents provide mechanisms for iterative code generation and revision, RIDGE contributes the finance-specific validation framework that guides this process and evaluates whether an implementation realizes the intended numerical operator and satisfies the mathematical principles underlying option pricing. The validation-driven framework further enables systematic computational method discovery when implementation refinement alone is no longer sufficient.

	The contributions of this paper are fourfold. First, we introduce RIDGE, an operator-based framework for validating generated option pricing implementations through deterministic measurement, diagnosis, and successive refinement. Second, we develop the notion of an evolving validation repository that accumulates expertise across validation iterations and transfers knowledge between model classes. Third, we demonstrate the framework on five stochastic-volatility models and show how implementation defects and inefficiencies can be detected and removed through successive validation iterations. Finally, we show that validation can drive methodological development, leading in certain cases to alternative pricing methodologies.
	
	The remainder of this paper is organized as follows. Section~\ref{sec:background} introduces the option-pricing setting and the operator formulation of RIDGE. Section~\ref{sec:procedure} presents the validation framework, including its measurement design, diagnostics, and repository refinements. Section~\ref{sec:experiments} reports the experimental studies, comprising both the replication experiments and the methodology-development cases, together with an evaluation of the resulting pricing methods. Section~\ref{sec:ablation} presents the ablation and reproducibility experiments that isolate the contribution of the validation operator. Section~\ref{sec:discussion} concludes with a discussion of the main lessons across models and directions for future work.

	\section{Validation of option pricing implementations}
	\label{sec:background}
	
	To study the proposed validation framework, we require a collection of asset price models and suitable option pricing methods that display a broad range of numerical behaviour. Stochastic-volatility models and their associated pricing techniques form a natural test bed for this purpose.

	The models considered here range from settings with analytic characteristic functions to models requiring substantially more involved numerical techniques. They  provide a natural environment in which generated implementations can be validated, refined, and, in some cases, replaced by alternative methodologies.
	We first introduce the stochastic-volatility models that provide the application domain for RIDGE and then present the framework itself, including its information flow, operator formulation, and the role of the language model in the iterative validation process.
	
	\subsection{Option valuation under stochastic volatility}
	
	The validation framework is developed and evaluated on a family of stochastic-volatility models of increasing complexity. These models are widely used in quantitative finance and provide a rich collection of numerical challenges, ranging from affine models with known characteristic functions to non-affine models for which valuation requires approximations or Monte Carlo simulation.
	
	Under the risk-neutral measure $\mathbb{Q}$, the asset price $S_t$ typically follows
	\begin{equation*}
		\frac{\mathrm{d}S_t}{S_t} = (r_t - q)\,\mathrm{d}t + \sqrt{v_t}\,\mathrm{d}W_t,
	\end{equation*}
	where $r_t$ denotes the short rate, $q$ the dividend yield, and $v_t$ a stochastic variance process. Depending on the model, the short rate may be constant or stochastic. In more general settings, jumps may be included~\cite{merton1976}:
	\begin{equation*}
		\frac{\mathrm{d}S_t}{S_t} = (r_t - q - \lambda m_J)\,\mathrm{d}t + \sqrt{v_t}\,\mathrm{d}W_t + (J_t-1)\,\mathrm{d}N_t,
	\end{equation*}
	where $N_t$ is a Poisson process with intensity $\lambda$, $J_t$ denotes the jump size, and $m_J$ is the jump compensator.
	
	The stochastic variance process may take different forms. In the Heston stochastic-volatility model~\cite{heston1993}, the variance satisfies
	\begin{equation*}
		\mathrm{d}v_t = \kappa(\theta - v_t)\,\mathrm{d}t + \sigma_v\sqrt{v_t}\,\mathrm{d}Z_t,\qquad
		\mathrm{d}[W,Z]_t = \rho\,\mathrm{d}t.
	\end{equation*}
	Here $v_t$ follows a Cox--Ingersoll--Ross square-root process~\cite{cox1985}. The Heston model belongs to the class of affine stochastic-volatility models, for which the logarithm of the asset price admits a characteristic-function representation. European option prices can then be computed efficiently by Fourier-based methods, like the Fourier-cosine (COS) method~\cite{fang2008}.
	
Many stochastic-volatility models considered in current research fall outside this affine setting. The five models used in this study span several such extensions of the Heston framework. The Bates model adds jumps while retaining an analytic characteristic function, whereas rough Heston replaces the Markovian variance dynamics by a Volterra process and requires the numerical solution of a fractional Riccati equation. The Heston--Hull--White model introduces stochastic interest rates, and a nonzero correlation between the equity and short-rate processes introduces a square-root dependence on the variance state, preventing the classical exponential--affine transform from closing under the standard
Riccati system. G-SVJD further replaces the square-root variance diffusion by a more general non-affine power-law specification. The Heston/Bates variance dynamics are recovered as a special case, whereas the more general specification is non-affine. These extensions lead to substantially different numerical valuation problems.

We defer the detailed model specifications and pricing methodologies to the corresponding validation studies in Section~\ref{sec:experiments}: Heston, Bates, and rough Heston are considered in Section~\ref{sec:rep_models}, while the Heston--Hull--White and G-SVJD models are treated in Section~\ref{sec:meth_dev}. For the latter two models, the validation process leads to alternative valuation methodologies. The purpose of the present section is to establish the common option-pricing setting in which the RIDGE framework is formulated.
	
	\subsection{Quantitative validation operator}\label{ssec:framework}
	
	RIDGE is an iterative framework in which pricing implementations evolve through successive iterations of measurement, diagnosis, and revision, while the expertise accumulated during validation is retained and reused. We first describe the information flow of a single RIDGE validation iteration and then formalize it in an operator-theoretic setting.
	
	Each iteration begins with three objects: a blueprint, a code-level pricing implementation, and a repository containing the accumulated validation knowledge. The blueprint outlines the model configurations and the implementation method for option pricing.  The repository documents methodologies and implementation details of the validation procedure,  plus model- and method-specific insights from previous validation experiences that support the interpretation of the resulting evidence.
	The implementation is then subjected to a collection of numerical measurements, producing an evidence record that is interpreted diagnostically. The resulting findings are used to revise the implementation and, if necessary, the blueprint itself, while the repository is updated with the expertise accumulated during the iteration. The validation process either terminates or continues from this revised state in the next iteration.
	Figure~\ref{fig:ridge-information-flow} summarizes this information flow.
	
	\begin{figure}[htb]
		\centering
		\includegraphics[width=0.99\textwidth]{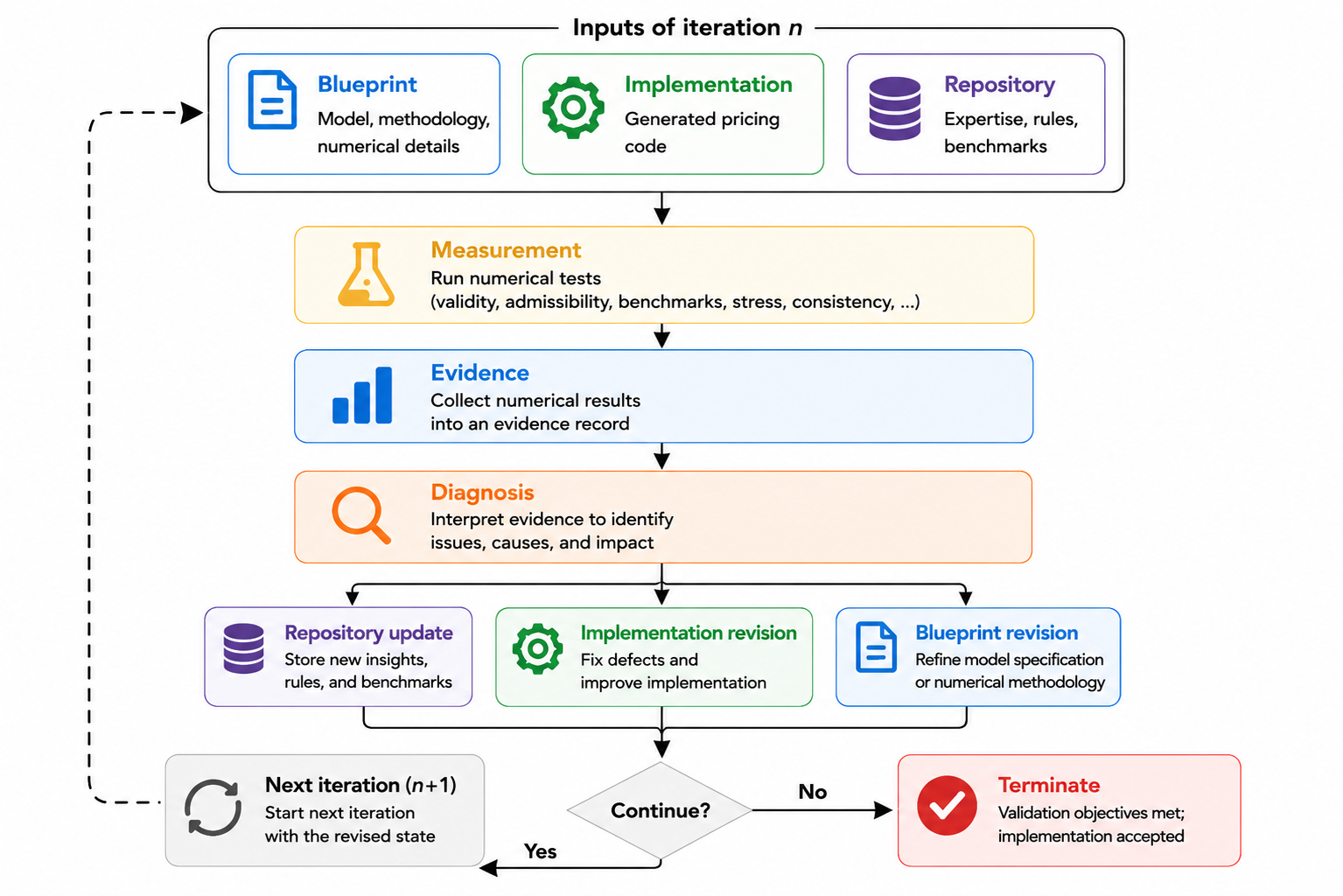}
		\caption{
			Information flow of a RIDGE validation iteration.
			A blueprint, pricing implementation, and repository define the current iteration.
			Measurement produces an evidence record, diagnosis interprets this evidence and updates both the repository and, when necessary, the blueprint and implementation.
			The revised state either initiates the next iteration or terminates the validation process.
		}
		\label{fig:ridge-information-flow}
	\end{figure}
	
	We now cast the validation framework RIDGE in an operator-theoretic setting. Let $\mathcal{M}$ denote the space of stochastic-volatility model specifications and let $\mathcal{G}$ denote the strike--maturity domain to be priced. The exact pricing operator is
	$$
	\mathcal{P}:\mathcal{M}\rightarrow\mathcal{Y},
	$$
	where $\mathcal{Y}$ is the space of admissible pricing surfaces over $\mathcal{G}$; each model specification $M\in\mathcal{M}$ is associated with its theoretical pricing surface $\mathcal{P}(M)$.
	
	A numerical method is characterized by a collection of numerical controls, denoted by $h$, which specify the truncation domain and the discretization parameters so as to balance accuracy against computational cost. These controls induce a finite representation of the pricing problem: a strike--maturity grid $\mathcal{G}_h\subset\mathcal{G}$, on which prices form vectors in a finite-dimensional space $\mathcal{Y}_h$, and a finite family of model specifications $\mathcal{M}_h\subseteq\mathcal{M}$ on which the implementation is evaluated.
	On this structure, the target pricing operator
	\[
	\mathcal{P}_h:\mathcal{M}_h\rightarrow\mathcal{Y}_h
	\]
	represents the ideal numerical realization of $\mathcal{P}$ under the controls $h$, whereas
	\[
	\widehat{\mathcal{P}}_h:\mathcal{M}_h\rightarrow\mathcal{Y}_h
	\]
	denotes the generated pricing implementation intended to approximate $\mathcal{P}_h$ (see Table~\ref{tab:operators}). For each model specification $M\in\mathcal{M}_h$, the vectors $\mathcal{P}_h(M)$ and $\widehat{\mathcal{P}}_h(M)$ contain the prices on the grid $\mathcal{G}_h$. All distances introduced below are evaluated over $\mathcal{M}_h$.
	
	\begin{table}[H]
		\centering
		\caption{Definitions of the pricing operators.}
		\label{tab:operators}
		\small
		\begin{tabular}{cl}
			\toprule
			\textbf{Operator} & \textbf{Meaning} \\
			\midrule
			$\mathcal{P}$ & exact pricing operator: the theoretical price profile of a model \\
			$\mathcal{P}_h$ & target pricing operator: a numerical realization of $\mathcal{P}$ with controls $h$ \\
			$\widehat{\mathcal{P}}_h$ & pricing implementation: the generated realization intended to approximate $\mathcal{P}_h$ \\
			$\mathcal{P}_h^{\mathrm{ref}}$ & reference value of $\mathcal{P}_h$, known on a subset of models \\
			\bottomrule
		\end{tabular}
	\end{table}
	
	Given two pricing operators $\mathcal{P}$ and $\mathcal{Q}$ and a model $M\in\mathcal{M}_h$, the \emph{model-wise distance} is
	\begin{equation*}
		d_M(\mathcal{P}, \mathcal{Q})
		\;:=\;
		\bigl\|\mathcal{P}(M)-\mathcal{Q}(M)\bigr\|_\infty,
	\end{equation*}
	the largest absolute price difference over the grid $\mathcal{G}_h$.
	Aggregated over the model family, we define the norm
	\[
	\|\mathcal{P}-\mathcal{Q}\|_{\mathcal{M}_h}
	:=
	\sup_{M\in\mathcal{M}_h}
	d_M(\mathcal{P},\mathcal{Q}).
	\]
	
	The total error of the implementation relative to the exact operator decomposes as
	\begin{equation}
		\label{eq:error_decomp}
		\|\widehat{\mathcal{P}}_h-\mathcal{P}\|_{\mathcal{M}_h}
		\;\le\;
		\underbrace{\|\widehat{\mathcal{P}}_h-\mathcal{P}_h\|_{\mathcal{M}_h}}_{\text{implementation error}}
		+
		\underbrace{\|\mathcal{P}_h-\mathcal{P}\|_{\mathcal{M}_h}}_{\text{scheme error}}.
	\end{equation}
	The first term is the object of validation: the aim is to reduce the implementation error so that the implementation correctly realizes the target operator $\mathcal{P}_h$. The second term reflects the approximation error of the numerical scheme and is treated separately.
	
The repository introduced above is denoted by $\mathcal{R}$. Its accumulated validation expertise is consulted by both the pricing implementation operator and the validation operator; dependence on the repository is indicated by the subscript $\mathcal{R}$.

The pricing implementation operator $\Phi_{\mathcal{R}}$ advances the iteration. It takes as input the blueprint $\mathcal{B}$, a structured specification of how the model is to be priced, comprising the stochastic model and its assumptions, the numerical methodology, and the implementation details.
	 The first iteration has no diagnosis to draw on and  starts from the source material, usually an article containing the model specification and a detailed methodology. This information is represented by the initial blueprint $\mathcal{B}^{(0)}$; in the simplest case, $\mathcal{B}^{(0)}$ contains the model specification and its assumptions. Each subsequent iteration then reads the current blueprint $\mathcal{B}^{(n)}$ together with the diagnostic report $\mathcal{D}^{(n)}$, which contains the findings of the previous validation iteration, namely the detected defects, their likely causes, and suggested revisions, and returns a revised blueprint and implementation,
	\[
	\bigl(\mathcal{B}^{(n+1)},\,\widehat{\mathcal{P}}_h^{(n+1)}\bigr)
	=
	\Phi_{\mathcal{R}}
	\!\left(
	\mathcal{B}^{(n)},
	\mathcal{D}^{(n)}
	\right),
	\qquad
	\mathcal{D}^{(0)}=\emptyset .
	\]
	
	The resulting implementation $\widehat{\mathcal{P}}_h^{(n+1)}$ need not be restricted to a single numerical method: $\Phi_{\mathcal{R}}$ may combine several numerical or simulation schemes within one implementation to cover the parameter and strike--maturity domains.
	
	The validation operator $\mathcal{V}_{\mathcal{R}}$ then evaluates $\widehat{\mathcal{P}}_h^{(n+1)}$ on the validation domain $\mathcal{M}_h$. It decomposes into measurement and diagnosis,
	\[
	\mathcal{V}_{\mathcal{R}}
	=
	\Gamma_{\mathcal{R}}
	\circ
	\mu_{\mathcal{R}}.
	\]
	
	The measurement operator $\mu_{\mathcal{R}}$ executes the relations prescribed by $\mathcal{R}$ on the implementation. It is organized into three blocks: validity, admissibility, and benchmarking, and returns a structured evidence record. The diagnosis operator $\Gamma_{\mathcal{R}}$ interprets this evidence and produces the next diagnostic report $\mathcal{D}^{(n+1)}$, a stop signal, and an updated repository for the next iteration.

	\noindent\textbf{Block~I (validity).}
	The implementation $\widehat{\mathcal{P}}_h$ emitted by $\Phi_\mathcal{R}$ must define a valid pricing operator. In particular, it must produce finite price vectors on $\mathcal{M}_h$, and the source code must correctly implement the methodology described in $\mathcal{B}$. A violation in this block terminates the iteration; the implementation is revised before any subsequent test is performed.
	
\noindent\textbf{Block~II (admissibility).}
For each model \(M\), let \(\mathcal A_M\subseteq\mathcal Y_h\) be the arbitrage-free set of price vectors on the grid \(\mathcal G_h\). For each \(M\in\mathcal M_h\), the admissibility error is
\[
A_M(\widehat{\mathcal P}_h)
:=
\inf_{U\in\mathcal A_M}
\bigl\|
\widehat{\mathcal P}_h(M)-U
\bigr\|_\infty .
\]
This quantity vanishes exactly when the price vector \(\widehat{\mathcal P}_h(M)\) is arbitrage-free on the validation grid. In practice, the admissibility of $\widehat{\mathcal P}_h(M)$ is assessed through the residuals of the parity, monotonicity, convexity, and density checks described in Section~\ref{sec:iteration}. The resulting residuals provide the operational admissibility measurements for each tested configuration and are used to assess deviation from $\mathcal A_M$.

Computational cost is measured separately as
$
\tau_M(\widehat{\mathcal P}_h),
$
and is used by the diagnosis stage to identify inefficient implementations and to prioritize numerical refinements. It is not part of the admissibility error.
	
	\noindent\textbf{Block~III (benchmarking).}
	The repository specifies a reference value $\mathcal{P}_h^{\mathrm{ref}}$ of the target operator $\mathcal{P}_h$, established on a subset $\mathcal{M}_{\mathrm{ref}}\subseteq\mathcal{M}_h$ of models. The reference may originate from benchmark prices documented in the raw material or from a validated and sound implementation. The agreement with this reference is measured, for each model $M\in\mathcal{M}_{\mathrm{ref}}$, by the concordance error
	\begin{equation}
		\label{eq:C_res}
		C_M(\widehat{\mathcal{P}}_h)
		:=
		d_M(\widehat{\mathcal{P}}_h,\mathcal{P}_h^{\mathrm{ref}}).
	\end{equation}
	Benchmarking is performed against the target operator rather than the exact one: $\mathcal{P}_h^{\mathrm{ref}}$ represents a trusted realization of the target discretization $\mathcal{P}_h$ on the configurations in $\mathcal{M}_{\mathrm{ref}}$. This provides the validation procedure with a concrete and attainable target against which implementations can be improved.

The measurement operator $\mu_\mathcal{R}: \widehat{\mathcal{P}}_h \mapsto \mathcal{E}$ gathers the outcomes of the three blocks into the structured evidence record
\begin{equation*}
	\mathcal{E}
	=
	\bigl(
	\mathcal{E}_{\mathrm{I}},
	\mathcal{E}_{\mathrm{II}},
	\mathcal{E}_{\mathrm{III}}
	\bigr),
	\qquad
	\mathcal{E}_{\mathrm{II}}
	=
	\bigl\{
	A_M(\widehat{\mathcal{P}}_h),
	\tau_M(\widehat{\mathcal{P}}_h)
	\bigr\}_{M\in\mathcal{M}_h},
	\qquad
	\mathcal{E}_{\mathrm{III}}
	=
	\bigl\{
	C_M(\widehat{\mathcal{P}}_h)
	\bigr\}_{M\in\mathcal{M}_{\mathrm{ref}}}.
\end{equation*}
Here, $\mathcal{E}_{\mathrm{I}}$ records the pass/fail outcome of the validity inspections in Block~I. The component $\mathcal{E}_{\mathrm{II}}$ collects the per-model admissibility measurements together with the corresponding computational costs, while $\mathcal{E}_{\mathrm{III}}$ collects the concordance measurements defined in~(\ref{eq:C_res}).

The \emph{diagnosis operator} $\Gamma_\mathcal{R}: \mathcal{E} \mapsto (\mathcal{D},S,\mathcal{R}')$ interprets the evidence and returns a diagnostic report $\mathcal{D}$, a stop signal $S$, and an updated repository $\mathcal{R}'$ for the next iteration. For each identified concern, whether related to the pricing implementation or the validation procedure, $\mathcal{D}$ records the finding, its diagnostic classification, and a suggested revision. Residuals that correspond to a pre-registered abnormal regime or to a documented limitation that persists after reasonable revision attempts may be classified as \texttt{[EXPECTED]}. Such residuals remain in the evidence record and are reported, but are no longer treated as actionable revision targets.

Accordingly, let
\[
A_M^{\mathrm{act}}(\widehat{\mathcal{P}}_h)
\qquad\text{and}\qquad
C_M^{\mathrm{act}}(\widehat{\mathcal{P}}_h)
\]
denote, respectively, the largest admissibility and concordance residuals for model $M$ that are not classified as \texttt{[EXPECTED]} by $\Gamma_\mathcal{R}$. The stop signal is set to \textit{halt} once all actionable admissibility and concordance errors fall below the prescribed tolerances $\eta_A,\eta_C>0$,
\begin{equation*}
S=
\begin{cases}
\textit{halt}, &
\displaystyle
\sup_{M\in\mathcal{M}_h}
A_M^{\mathrm{act}}(\widehat{\mathcal{P}}_h)
\le \eta_A
\quad\text{and}\quad
\sup_{M\in\mathcal{M}_{\mathrm{ref}}}
C_M^{\mathrm{act}}(\widehat{\mathcal{P}}_h)
\le \eta_C,
\\[3pt]
\textit{continue}, & \text{otherwise}.
\end{cases}
\end{equation*}
Thus, termination does not require every measured residual to vanish or fall below tolerance; it requires that no discrepancy above tolerance remains an actionable revision target.

The resulting iteration is
\begin{equation*}
	\bigl(\mathcal{B}^{(n)},\,\mathcal{D}^{(n)}\bigr)
	\;\xrightarrow{\,\Phi_{\mathcal{R}^{(n)}}\,}\;
	\bigl(\mathcal{B}^{(n+1)},\,\widehat{\mathcal{P}}_h^{(n+1)}\bigr)
	\;\xrightarrow{\,\mu_{\mathcal{R}^{(n)}}\,}\;
	\mathcal{E}^{(n+1)}
	\;\xrightarrow{\,\Gamma_{\mathcal{R}^{(n)}}\,}\;
	\bigl(\mathcal{D}^{(n+1)},\,S^{(n+1)},\,\mathcal{R}^{(n+1)}\bigr),
\end{equation*}
halting at the first iteration for which $S=\textit{halt}$. Through this loop the implementation $\widehat{\mathcal{P}}_h$ is guided toward the target operator $\mathcal{P}_h$, with the aim of reducing the implementation-error term in decomposition~(\ref{eq:error_decomp}) until no actionable discrepancy remains above tolerance.

The repository $\mathcal{R}$ is not fixed: each diagnosis returns an updated $\mathcal{R}'$, so $\mathcal{R}$ evolves from iteration to iteration and across the models it has seen. Unlike a standard LLM coding agent, which evolves source code, RIDGE evolves validated domain knowledge that subsequently guides implementation refinement. Its evolution has two aspects: the accumulation of model- and method-specific expertise and the refinement of the measurement and diagnosis procedures themselves in response to validation feedback. Because structural relations are accumulated across model classes, insights gained during the validation of one model can refine the strategy applied to subsequent ones. Beyond refining a given method, this accumulation may also reveal the need for a better-suited one, as for the HHW and G-SVJD models of Section~\ref{sec:meth_dev}.
	
	\subsection{Stochastic dynamics of validation}\label{ssec:stochastic}

	The preceding operator formulation describes a deterministic measurement layer coupled to a stochastic implementation generator. We now give a conditional convergence interpretation of this interaction. The purpose is not to assert convergence for an arbitrary language model, but to identify sufficient conditions under which repository-driven diagnosis produces a validation trajectory approaching the stopping region defined above.

	Let $(\Omega,\mathcal{F},\mathbb{P})$ be a probability space that describes the randomness in the language-model generation and diagnosis stages, and let $\{\mathcal{F}_n\}_{n\geq0}$ denote the associated filtration, with $\mathcal{F}_n$ generated by the blueprints, implementations, repositories, evidence, and diagnostic reports observed up to iteration $n$. The iteration state is
	\[
	X_n
	=
	\bigl(
	\mathcal{B}^{(n)},
	\widehat{\mathcal{P}}_h^{(n)},
	\mathcal{R}^{(n)}
	\bigr).
	\]
Conditional on $\mathcal{F}_n$, the generation operator proposes the next blueprint and implementation through the stochastic language-model generation step. The resulting implementation is then evaluated according to the measurement protocol prescribed by $R^{(n)}$. Thus, randomness enters only through the proposed revision. Once the generated implementation and the corresponding validation protocol have been fixed, the Block~II and Block~III measurements are deterministic.

For iterations that pass the Block~I validity inspection, define the tolerance-adjusted validation energy
\begin{equation}
L_n :=
\left[
\sup_{M\in\mathcal M_h}
A_M^{\mathrm{act}}
\!\left(\widehat{\mathcal P}^{(n)}_h\right)
-\eta_A
\right]_+^2
+
\left[
\sup_{M\in\mathcal M_{\mathrm{ref}}}
C_M^{\mathrm{act}}
\!\left(\widehat{\mathcal P}^{(n)}_h\right)
-\eta_C
\right]_+^2,
\label{eq:validation_energy}
\end{equation}
where $[x]_+ := \max\{x,0\}$. A Block~I failure is handled by immediate diagnosis and revision, as specified above, and is excluded from this quantitative energy. By construction, $L_n=0$ precisely when the actionable admissibility and concordance conditions defining $S=\textit{halt}$ are satisfied. The zero set of $L_n$ is consequently the validation-defined acceptance region. 

Updates to the repository may temporarily increase the measured validation energy of an implementation that previously satisfied the current validation criteria. Another source of disturbance in performance comes from the stochastic generation of a new blueprint and implementation, where side effects may be introduced while repairing previously diagnosed deficiencies. We represent these two effects by nonnegative $\mathcal{F}_n$-measurable sequences $\alpha_n$ and $\beta_n$, respectively. The corrective effect of validation is represented by a nonnegative dissipation term $\gamma_n$.

	\begin{proposition}[Conditional convergence of the validation energy]
		\label{prop:validation_convergence}
		Suppose that the nonnegative adapted process $(L_n)_{n\geq0}$ satisfies
		\begin{equation}
			\label{eq:rs_validation}
			\mathbb{E}\!\left[L_{n+1}\mid\mathcal{F}_n\right]
			\leq
			(1+\alpha_n)L_n-\gamma_n+\beta_n,\qquad\text{a.s.,}
		\end{equation}
		where
		\[
		\sum_{n=0}^{\infty}\alpha_n<\infty,
		\qquad
		\sum_{n=0}^{\infty}\beta_n<\infty
		\qquad\text{a.s.}
		\]
		Here a.s. denotes almost sure convergence. Assume further that there exists a continuous function $c:[0,\infty)\to[0,\infty)$ with $c(0)=0$ and $c(\ell)>0$ for every $\ell>0$ such that
		\[
		\gamma_n\geq c(L_n)
		\qquad\text{a.s.}
		\]
		Then
		\[
		L_n\longrightarrow0
		\qquad\text{a.s.}
		\]
	\end{proposition}

	\begin{proof}
		The Robbins--Siegmund almost-supermartingale convergence theorem~\cite{RobbinsSiegmund1971}, applied to \eqref{eq:rs_validation}, implies that $L_n$ converges almost surely to a finite random variable $L_\infty\geq0$ and that $\sum_{n=0}^{\infty}\gamma_n<\infty$ almost surely. Since $\gamma_n\geq c(L_n)$, it follows that $c(L_n)\to0$ almost surely. Continuity of $c$ and the almost-sure convergence $L_n\to L_\infty$ give $c(L_\infty)=0$. The strict positivity of $c$ away from the origin yields $L_\infty=0$ almost surely.
	\end{proof}

	Proposition~\ref{prop:validation_convergence} separates the role of validation from the intrinsic capability of the language model. The result does not require deterministic generation. Instead, it requires that repository expansion and residual generative perturbations have finite cumulative effect, and that away from the acceptance region the diagnostic feedback induces strictly positive expected dissipation. Under these assumptions, the validation energy converges almost surely to the region in which the prescribed actionable admissibility and concordance tolerances are satisfied. The proposition is an asymptotic result: it does not imply that the stopping region is reached in finitely many iterations. Establishing finite-time guarantees, or verifying the drift conditions for specific classes of language-model generators, requires additional analysis.

	\subsection{The role of LLM}
	
	The language model enters the framework in two capacities. First, it is generative: from a model specification it produces a methodological blueprint together with an executable pricing implementation, through the generation operator $\Phi_\mathcal{R}$. Second, it is evaluative: it contributes to the Block~I validity inspections within the measurement $\mu_\mathcal{R}$ and to the diagnosis operator $\Gamma_\mathcal{R}$.
	
	The framework deliberately separates language-model reasoning from numerical measurement. Once the implementation, the validation family $\mathcal{M}_h$, and the repository are fixed, the admissibility and benchmarking evidence is generated deterministically from numerical calculation, so the figures and flags of the numerical checks are reproducible. Language-model reasoning is confined to three stages: the generation operator $\Phi_\mathcal{R}$, the Block~I inspections, and the diagnosis operator $\Gamma_\mathcal{R}$.
	
	The quality of the resulting implementation depends on the capability of the language model. Stronger models may identify more effective revisions or discover alternative formulations that weaker ones overlook, as illustrated later by the HHW and G-SVJD case studies. The stochastic interpretation in Section~\ref{ssec:stochastic} thus does not assert convergence for an arbitrary generator. Rather, Proposition~\ref{prop:validation_convergence} identifies sufficient drift and perturbation conditions under which repository-driven refinement approaches the validation-defined acceptance region. The reproducibility experiments of Section~\ref{sec:ablation} provide empirical evidence consistent with this interpretation: once numerical measurement anchors the validation process, generators of comparable capability can be guided to validated implementations.
\section{Implementation of the validation operator}\label{sec:procedure}

This section describes how the validation framework of \S\ref{ssec:framework} is realized in practice. We focus on the measurement and diagnosis procedures that emerged from the experiments and were subsequently incorporated into the repository.

\subsection{Measurement design based on expert knowledge}\label{sec:iteration}

The repository guides the definition of the validation domain and the measurements performed during each iteration. Its accumulated experience identifies challenging parameter regimes and informative diagnostic tests. The measurement procedure can thus evolve as validation experience is transferred across model classes.

Each validation iteration begins with the validity assessment of Block~I. The blueprint is checked against the underlying mathematical sources, the generated implementation is compared with the methodology specified in the blueprint, and the resulting code is required to execute successfully with finite and well-defined outputs throughout the validation domain. Failure of any of these checks terminates the current iteration and produces a diagnostic report; valid implementations proceed to Block~II.

\paragraph{Validation domain.}

The admissibility checks are performed on a finite strike--maturity domain covering typical operating regimes and numerically challenging configurations. The strike range at each maturity is chosen according to the characteristic scale of the underlying log-return distribution: short maturities are concentrated near the money, while longer maturities admit wider strike ranges. The resulting configurations define the validation grid $\mathcal G_h$. The maturity grid and strike-selection procedure are specified in Section~\ref{sec:settings}.

\paragraph{Stress-test design.}

Alongside the default parameter configuration, the implementation is evaluated on stressed parameter sets. Their design is model-agnostic and based on the local sensitivity of the at-the-money implied volatility and implied-volatility skew at selected maturities. Finite differences identify parameter perturbations that produce substantial changes while remaining within the admissible parameter region.

Representative positive and negative perturbations are ranked by their impact, and a limited number of informative configurations is retained. Model-specific abnormal regimes documented in the literature, such as violations of the Feller condition in Heston-type models, are added when applicable. The default, stressed, and abnormal configurations together constitute the validation family $\mathcal M_h$. Further details of the stress-set design are given in Appendix~\ref{app:stress_design}.

\paragraph{Consistency checks.}

Block~II evaluates four consistency properties: put--call parity, monotonicity, convexity, and the implied risk-neutral density. For each property, a residual is measured and compared with a prescribed tolerance. Residuals exceeding the tolerance are recorded as violations and passed to the diagnosis stage.

The put--call parity tests an analytical identity and is particularly sensitive to implementation errors. Monotonicity checks the no-arbitrage ordering of call prices in  strike, while convexity is assessed through finite-difference approximations of butterfly spreads. The latter two reveal local inconsistencies in the pricing surface. Finally, the implied risk-neutral density is recovered through the Breeden--Litzenberger relation~\cite{breeden1978} and checked for non-negativity and normalization, exposing more global deficiencies in the numerical approximation. The worst residual over $\mathcal G_h$ determines the admissibility error $A_M$ for each tested configuration.

\paragraph{Benchmark checks.}

Block~III compares the implementation with a reference value of the target operator $\mathcal P_h^{\mathrm{ref}}$. Reference sources follow a fixed hierarchy. Published numerical benchmarks are used when available, followed by known degeneration limits and, if neither is available, the local affine approximation of Appendix~\ref{app:affine}. Numerical refinement provides a final internal consistency check. Only the highest-priority applicable reference is used.

The worst discrepancy over $\mathcal G_h$ defines the concordance error $C_M$. Together with the admissibility measurements, these results form the numerical evidence supplied to the diagnosis stage. The accuracy grade and its interpretation for the different model and methodology classes are described in Section~\ref{sec:settings}.

Table~\ref{tab:procedure} summarizes the practical realization of the three validation blocks.

The validation studies also led to several practical refinements of the measurement procedure. To apply the same validation operator across different model classes, each generated implementation is wrapped in a standardized adapter providing a common representation of prices, parameters, diagnostics, and computational statistics.

The definition of the validation grid was refined during the experiments. A fixed moneyness range places short-maturity contracts far into the tails of the log-return distribution, where numerical truncation effects may dominate the validation results. The regular grid uses a maturity-dependent strike window whose width scales with the characteristic volatility of the horizon log-return. Since such a restriction may hide instabilities confined to extreme regions, the grid is supplemented by targeted boundary tests at very short and long maturities and at extreme moneyness levels. Their precise specification is given in Appendix~\ref{app:stress_design}.

Finally, the evidence record stores the magnitude of each violation rather than only whether a prescribed tolerance is exceeded. This allows the diagnosis stage to distinguish isolated minor deficiencies from systematic shortcomings and to prioritize subsequent revisions.

\begin{table}[H]
	\centering
	\caption{\revthree{The three validation blocks of RIDGE and their role in measurement and diagnosis.}}
	\label{tab:procedure}
	\small
	\renewcommand{\arraystretch}{1.15}
	\begin{tabular}{p{2.8cm} p{3.6cm} p{3.4cm} p{3.4cm}}
		\toprule
		\textbf{Step} & \textbf{Procedure} & \textbf{Purpose} & \textbf{Consequence of failure} \\
		\midrule
		\multicolumn{4}{l}{\revthree{\textit{Block~I (validity): does $\widehat{\mathcal{P}}_h$ constitute a valid pricing implementation?}}} \\
		\midrule
		\revthree{I.1 Blueprint inspection} &
		\revthree{Compare $\mathcal{B}$ with cited sources; check internal consistency of model, methodology, and implementation details} &
		\revthree{Ensure the methodology to be implemented is mathematically coherent} &
		\revthree{Iteration terminates; $\mathcal{B}$ or $\Phi_\mathcal{R}$ revised} \\

		\revthree{I.2 Implementation inspection} &
		\revthree{Match source code against the methodology of $\mathcal{B}$} &
		\revthree{Ensure $\widehat{\mathcal{P}}_h$ implements what $\mathcal{B}$ specifies} &
		\revthree{Iteration terminates; implementation revised} \\

		\revthree{I.3 Execution assessment} &
		\revthree{Execute the implementation and check that all outputs are finite and free of failures} &
		\revthree{Ensure $\widehat{\mathcal{P}}_h$ executes successfully across $\mathcal{M}_h$} &
		\revthree{Iteration terminates; implementation revised} \\
		\midrule

		\multicolumn{4}{l}{\revthree{\textit{Block~II (admissibility): does $\widehat{\mathcal{P}}_h(M)$ lie in the arbitrage-free set of price vectors $\mathcal{A}_M$?}}} \\
		\midrule
		\revthree{II.1 Domain definition} &
		\revthree{Generate the validation grid and stress-test parameter sets} &
		\revthree{Construct the validation family $\mathcal{M}_h$ and fix the grid $\mathcal{G}_h$} &
		None \\

		\revthree{II.2 Consistency checks} &
		Parity, monotonicity, convexity, and density checks over $\mathcal{M}_h$ &
		Measure the admissibility error $A_M(\widehat{\mathcal{P}}_h)$ &
		\revthree{Admissibility error $A_M$ recorded in $\mathcal{E}_{\mathrm{II}}$} \\
		\midrule

		\multicolumn{4}{l}{\revthree{\textit{Block~III (benchmarking): does $\widehat{\mathcal{P}}_h$ agree with the available reference evidence?}}} \\
		\midrule
		\revthree{III.1 Reference comparison} &
		Compare against the highest-priority available reference &
		Measure the concordance error $C_M(\widehat{\mathcal{P}}_h)$ &
		\revthree{Concordance error $C_M$ recorded in $\mathcal{E}_{\mathrm{III}}$} \\
		\midrule

		\multicolumn{4}{p{13.6cm}}{
		\textit{Diagnosis:}
		a Block~I failure is diagnosed immediately and the iteration terminates.
		After Blocks~II and III, the numerical evidence
		$(\mathcal{E}_{\mathrm{II}},\mathcal{E}_{\mathrm{III}})$
		is interpreted by $\Gamma_{\mathcal R}$, which returns the diagnostic report $\mathcal D$, the updated repository $\mathcal R'$, and the stop signal $S$.
		} \\
		\bottomrule
	\end{tabular}
\end{table}

	\subsection{Iteration diagnostics}
	
	The diagnostic report contains one entry for each relevant finding identified during a validation iteration. Each entry records the targeted component, whether the pricing implementation or the repository design, a diagnostic category drawn from a fixed taxonomy, a description of the observed issue, and a suggested corrective action. The purpose of the record is to identify the most plausible causes, to make the state of the current iteration explicit, and to guide the next revision.
	
	\begin{table}[H]
		\centering
		\caption{\reveight{Diagnostic categories used by the validation operator.}}
		\label{tab:tags}
		\small
		\renewcommand{\arraystretch}{1.15}
		\begin{tabular}{p{3.0cm}p{10.2cm}}
			\toprule
			\textbf{Category} & \textbf{Meaning} \\
			\midrule
			\texttt{[BUG]} &
			Defects that prevent the implementation from correctly realizing the intended numerical method. \\
            \texttt{[CONSISTENCY]} &
			Violations of admissibility requirements, including parity, monotonicity, convexity, and density consistency. \\
			\texttt{[METHOD-ALGO]} &
			Possible improvements within the chosen numerical methodology or its implementation. \\
			\texttt{[METHOD-HYPER]} &
			Possible improvements in numerical controls selection, adaptivity, or tuning rules. \\
			\texttt{[METHOD-CHOICE]} &
			Indications that a different numerical methodology may be better suited for the model class under consideration. \\
			\texttt{[EXPECTED]} &
			Residual limitations attributed to the model, benchmark, or approximation framework that survive after several validation iterations. \\
			\bottomrule
		\end{tabular}
	\end{table}
	
	The diagnostic categories serve different purposes within the validation iteration. Validity failures are classified as \texttt{[BUG]}, where fatal ones cause an immediate halt and must be resolved in another iteration. Admissibility violations are typically classified as \texttt{[CONSISTENCY]}, while benchmarking discrepancies may indicate implementation shortcomings (\texttt{[METHOD-ALGO]}), hyperparameter choices (\texttt{[METHOD-HYPER]}), or limitations of the numerical methodology itself (\texttt{[METHOD-CHOICE]}). Persistent deficiencies that remain after several revision iterations may ultimately indicate that an inherent limitation of the current implementation, leading to an \texttt{[EXPECTED]} diagnosis. The full classification is summarized in Table~\ref{tab:tags}.
	
	Alongside the per-finding entries, the diagnosis assigns a heuristic score to each category. The scores are calibrated to be broadly comparable across categories and iterations and provide a compact summary of the implementation state. They are intended to prioritize revisions and to track progress across validation iterations; they do not constitute validation criteria themselves.
	
	The stop signal \(S\) implements the tolerance criterion introduced in Section~2.2. In practice, this criterion is interpreted as indicating that no remaining discrepancy can be removed by a reasonable revision of the implementation. Residual deficiencies that cannot be eliminated are  ultimately transferred to the category \texttt{[EXPECTED]}, either because the corresponding regime was declared abnormal in the blueprint prior to any measurement or because repeated revision attempts have failed and a theoretical explanation has been recorded.
	
	When \(S=\textit{continue}\), the diagnostic report influences the next iteration in two ways. First, it guides the revision of the pricing implementation through the operator \(\Phi_{\mathcal R}\). Second, it contributes to the refinement of the repository by recording newly observed failure modes and effective corrective actions. 

The validation repository stores structured numerical expertise rather than source code. Entries are maintained as validation progresses, with explicit status tracking for attempted resolutions. Table~\ref{tab:repo_categories} summarizes the main entry categories and their characteristics.

\begin{table}[ht]
\centering
\caption{Repository entry categories and their characteristics.}
\label{tab:repo_categories}
\begin{tabular}{p{0.18\linewidth} p{0.38\linewidth} p{0.38\linewidth}}
\toprule
\textbf{Category} & \textbf{Description} & \textbf{Comments} \\
\midrule
Model-specific action & Records a defect identified during validation of a particular model, and the corrective action taken. & Traced across iterations with outcome and status. May involve multiple attempts. \\
\addlinespace
Method-specific action & Records an algorithmic or methodological implementation improvement tailored to specific parameter regimes. & Traced with outcome and status. Parameter-dependent rules discovered and refined across iterations. \\
\addlinespace
Generic principle & Records a best-practice experience that applies broadly across a class of models or numerical methods. & Inherited expertise; stored as reusable guidance. \\
\addlinespace
Model linkage & Records inheritance relationship from prior validation tasks. & e.g., Heston $\to$ Bates $\to$ Rough Heston. \\
\bottomrule
\end{tabular}
\end{table}

\begin{exam}[Typical repository entries]
\rm During the validation study of the Heston and Bates models, the following repository entries are created.
\item \textit{Model-specific action (Heston):}
  \begin{itemize}[nosep]
    \item \textit{Observed validation evidence:} Errors in deep Feller-violation regimes for long maturities.
    \item \textit{Description:} The standard domain and resolution fail for extreme parameter sets.
    \item \textit{Corrective action:} Widen the domain and adapt resolution in Feller-violating regimes.
    \item \textit{Evolution} \& \textit{status:} Three attempts across Iteration 7--9; final resolution used smooth parity blending. Status: \textit{resolved}.
  \end{itemize}

\item \textit{Method-specific action (Bates):}
  \begin{itemize}[nosep]
    \item \textit{Observed validation evidence:} Overflow in call coefficients for extreme jump parameters.
    \item \textit{Description:} Call COS coefficients may overflow when the domain is wide.
    \item \textit{Corrective action:} Price the put side when the domain is wide and recover calls via put-call parity.
    \item \textit{Evolution} \& \textit{status:}  The first attempt was insufficient; replaced by a more robust conditioning strategy that works for all wide-domain cases. Status: \textit{resolved}.
  \end{itemize}

\item \textit{Generic principle:}
  \begin{enumerate}[nosep]
    \item The COS truncation domain must reflect the distribution width at maturity, not merely the initial variance scaled by time.
    \item When direct computation is unstable, compute the stable option price and recover the desired option value using put-call parity.
  \end{enumerate}
  Each validation iteration queries the repository using the current model class, numerical methodology, and observed diagnostic evidence; relevant entries are retrieved and applied to guide implementation refinement. Successful resolutions accumulate over time, while failed attempts are logged to prevent repeated effort. This allows expertise to be transferred across model classes and validation iterations without requiring the validation agent to rediscover known fixes.
\end{exam}

\section{Experiments}\label{sec:experiments}

We apply RIDGE to five stochastic-volatility models in two types of validation study. The first concerns Heston, Bates, and rough Heston, for which established pricing methodologies are available. Heston combines the analytic characteristic function of~\cite{heston1993} with the Fourier-cosine (COS) method~\cite{fang2008}. Bates extends the characteristic function with a jump component~\cite{bates1996}, while rough Heston obtains the transform numerically from the fractional Riccati equation of El~Euch and Rosenbaum~\cite{el2019rough}; both are subsequently priced by the COS method as well. For these three models, the validation task is to identify and remove defects in generated implementations of documented pricing methodologies. Section~\ref{sec:rep_models} follows their development from the initial generated implementations to their validated forms.

The second type of study concerns the Heston--Hull--White (HHW) and G-SVJD models. In these cases, validation indicates limitations of the initially generated Monte Carlo methodologies that cannot be resolved through implementation refinement alone. The diagnostic process therefore indicates alternative semi-analytic CCF-PDE formulations, developed and validated in Section~\ref{sec:meth_dev}. Section~\ref{sec:performance} compares the resulting implementations in terms of accuracy, robustness, and computational cost.

\subsection{\texorpdfstring{Experiment settings}{Experiment settings}}
\label{sec:settings}

All five validation studies use Claude Opus~4.8 as the language model responsible for the three model-driven stages of RIDGE: blueprint-and-implementation generation, the Block~I inspections, and diagnostic interpretation. The experiments use a common validation domain and measurement protocol.

The maturity grid is
\[
T \in \{0.001,\,0.004,\,0.01,\,0.04,\,0.1,\,0.25,\,0.5,\,1.0,\,2.0,\,5.0\},
\]
and the candidate moneyness grid is
\[
\log(K/F) \in [-0.5, 0.5],
\]
 where \(K\) is equally and densely spaced, and \(F\) denotes the forward price at maturity \(T\). The concentration of maturities at the short end targets the region where numerical challenges are typically most pronounced, while the candidate moneyness grid extends to log-moneyness values of \(\pm0.5\), covering far in- and out-of-the-money regions.

At each maturity, the selected strikes are restricted using the scale of the COS truncation range,
\[
|\log(K/F)|
\le
\tfrac12 L \sqrt{\bar v},
\qquad L=14,
\]
where
\[
\bar v
=
\theta T
+
(v_0-\theta)
\frac{1-e^{-\kappa T}}{\kappa}
\]
is the expected integrated variance to maturity and reduces to \(v_0T\) in the absence of mean reversion. The resulting strike--maturity configurations define the validation grid \(\mathcal G_h\). The stress configurations and boundary tests used alongside the default parameter set are specified in Appendix~\ref{app:stress_design}.

The measurements record four consistency residuals at each selected point. Writing \(V_c(T,K)\) and \(V_p(T,K)\) for the call and put prices, respectively, and \(P(0,T)\) for the maturity-\(T\) discount factor, the implied forward density is recovered through the Breeden--Litzenberger relation
\[
f_{\!Q}(T,K)=\frac{\partial_K^2 V_c(T,K)}{P(0,T)}.
\]

The parity residual is defined as the relative put--call parity error,
\[
e_1=
\frac{
	\bigl|V_c-V_p-(F-K)P(0,T)\bigr|
}
{K\,P(0,T)}.
\]

The monotonicity residual is the worst no-arbitrage ordering violation over strike,
\[
e_2
=
\max_{i, j}
[\,V_c(T_i,K_{j+1})-V_c(T_i,K_j)\,]_+.
\]

The convexity residual is the worst negative strike butterfly,
\[
e_3
=
\max_{i, j}
[\, -(V_c(T_i,K_{j-1})
-2V_c(T_i,K_j)
+V_c(T_i,K_{j+1}))\,]_+,
\]
while the density residual measures the deviation of the recovered density from unit mass,
\[
e_4
=
\left|
\int f_{\!Q}(T,K)\,\mathrm dK
-1
\right|.
\]

A configuration is flagged whenever a residual exceeds its prescribed tolerance. The tolerances are fixed across models. In particular, the parity tolerance is \(10^{-5}\) for deterministic transform implementations and is relaxed to \(10^{-3}\) for implementations subject to intrinsic Monte Carlo error.

Benchmarking compares the implementation with a reference solution. The benchmark deviation \(\Delta\mathrm{IV}\) is the maximum absolute at-the-money implied-volatility discrepancy and is mapped to a letter grade according to
\[
\lambda(\Delta\mathrm{IV})
=
\begin{cases}
	\mathrm{D}, & \Delta\mathrm{IV}<0.001\%,\\
	\mathrm{C}, & 0.001\%\le \Delta\mathrm{IV}<0.01\%,\\
	\mathrm{B}, & 0.01\%\le \Delta\mathrm{IV}<0.10\%,\\
	\mathrm{A}, & 0.10\%\le \Delta\mathrm{IV}<1.00\%,\\
	\mathrm{F}, & \Delta\mathrm{IV}\ge 1.00\%.
\end{cases}
\]

The diagnosis interprets these grades relative to the numerical complexity of the model and the pricing methodology rather than against a universal acceptance threshold. Lower grades may be tolerated for challenging models, for example those with slow mean reversion or jumps, and for Monte Carlo implementations, whose slower convergence generally makes the highest accuracy grades more challenging to attain than for deterministic methods.

Based on this evidence, the diagnosis assigns heuristic scores to the diagnostic categories. The scores prioritize revisions and track progress across validation iterations; they are not validation criteria. When the documented methodology is already established, the five categories \texttt{[BUG]}, \texttt{[METHOD-ALGO]}, \texttt{[METHOD-HYPER]}, \texttt{[CONSISTENCY]}, and \texttt{[EXPECTED]} are each scored out of \(20\), giving an iteration total out of \(100\). If a \texttt{[METHOD-CHOICE]} diagnosis arises during validation, this category is additionally scored out of \(20\) and contributes to the total.

The reported validation scores should be interpreted as heuristic summaries of the accumulated validation evidence rather than as quantitative measures of implementation accuracy. They are derived from the deterministic validation residuals, admissibility tests, and benchmark comparisons described above. They summarize validation progress across successive validation iterations, while the underlying validation evidence remains the primary basis for assessing implementation quality.

A residual that cannot be eliminated is eventually transferred to \texttt{[EXPECTED]}. This occurs when the corresponding regime was declared abnormal in the blueprint before measurement, or when repeated independent revision attempts fail to remove the discrepancy and a theoretical explanation has been recorded in the repository.
	
\subsection{Replication runs: Heston, Bates, and rough Heston}\label{sec:rep_models}

To study how RIDGE validates and refines existing pricing methodologies, we consider three documented methods of increasing numerical complexity, for the Heston, Bates, and rough Heston models.

The three pricing methodologies share the COS valuation step and differ primarily in the derivation of the characteristic function. Heston evaluates an analytic transform, Bates augments this transform with a jump component, and rough Heston computes the transform numerically by solving a fractional Riccati equation. For Heston and Bates, the numerical approximation of choice is the COS valuation once the characteristic function has been specified. For rough Heston, the characteristic function supplied to the COS step is itself computed numerically by a Volterra solver with its own discretization and stability properties.

Table~\ref{tab:tier_progression} records the main implementation revision at each iteration together with the corresponding diagnostic scores. Each run starts from an unadapted implementation of the documented methodology and ends with a validated one: Heston improves over ten iterations from \(31/100\) to \(97/100\), Bates over five iterations from \(58/100\) to \(96/100\), and rough Heston over six iterations from \(32/100\) to \(99/100\). The scores serve as progress indicators and identify the diagnostic categories that still limit the implementation; the underlying evidence is provided by the deterministic measurements of Section~\ref{sec:iteration}.
	
	\begin{table}[H]
		\centering
		\caption{\revfive{Per-iteration category scores (column headings abbreviate the tags of Table~\ref{tab:tags}: [ALGO]$=$[METHOD-ALGO], [HYP]$=$[METHOD-HYPER], \reveight{[CONS]$=$[CONSISTENCY],} and [EXP]$=$[EXPECTED]). The scores summarize the diagnostic interpretation of the measured evidence; the total is the sum of [BUG], [ALGO], [HYP], [CONS], and [EXP].}}
		\label{tab:tier_progression}
		\footnotesize
		\setlength{\tabcolsep}{2pt}
		\begin{tabular}{cclcccccr}
			\toprule
			\textbf{Model} & \textbf{Iter.} & \textbf{Principal refinement} & \texttt{[BUG]} & \texttt{[ALGO]} & \texttt{[HYP]} & \texttt{[CONS]} & \texttt{[EXP]} & \textbf{Total} \\
			\midrule
			\multirow{10}{*}{Heston}
			& 1 & analytic ChF; naive COS domain & 2  & 4  & 5  & 0  & 20 & 31/100 \\
			& 2 & payoff endpoints; cumulant domain & 6  & 4  & 5  & 2  & 20 & 37/100 \\
			& 3 & put domain; strike coverage & 14 & 4  & 8  & 4  & 20 & 50/100 \\
			& 4 & vectorization; adaptive resolution & 12 & 16 & 8  & 8  & 20 & 64/100 \\
			& 5 & wrapper passes adaptive \(N\) & 16 & 14 & 16 & 12 & 20 & 78/100 \\
			& 6 & out-of-the-money valuation & 14 & 16 & 16 & 12 & 20 & 78/100 \\
			& 7 & batched OTM valuation & 20 & 18 & 14 & 14 & 20 & 86/100 \\
			& 8 & Feller-aware domain (1) & 20 & 18 & 17 & 16 & 20 & 91/100 \\
			& 9 & Feller-aware domain (2) & 20 & 18 & 18 & 17 & 20 & 93/100 \\
			& 10 & smooth put--call blending & 20 & 20 & 20 & 20 & 17 & 97/100 \\
			\midrule
			\multirow{5}{*}{Bates}
			& 1 & jump-augmented ChF; COS & 12 & 10 & 6  & 10 & 20 & 58/100 \\
			& 2 & inherited Heston refinements & 20 & 18 & 18 & 17 & 20 & 93/100 \\
			& 3 & banded put--call blending; put estimator & 20 & 19 & 19 & 18 & 20 & 96/100 \\
			& 4 & fourth-cumulant domain & 20 & 20 & 20 & 19 & 20 & 99/100 \\
			& 5 & validator quadrature corrected & 20 & 20 & 20 & 20 & 16 & 96/100 \\
			\midrule
			\multirow{6}{*}{\shortstack{Rough\\Heston}}
			& 1 & fractional-Riccati ChF; COS & 2  & 4  & 4  & 2  & 20 & 32/100 \\
			& 2 & fractional Adams solver & 6  & 10 & 4  & 2  & 20 & 42/100 \\
			& 3 & stability-based step count & 14 & 14 & 5  & 4  & 20 & 57/100 \\
			& 4 & inherited COS refinements & 18 & 15 & 15 & 14 & 20 & 82/100 \\
			& 5 & implicit fractional corrector & 18 & 18 & 18 & 16 & 20 & 90/100 \\
			& 6 & narrower parity band & 20 & 20 & 20 & 20 & 19 & 99/100 \\
			\bottomrule
		\end{tabular}
	\end{table}
The Heston replication run is the first validation study and begins without inherited expertise. Most of the numerical machinery later reused by Bates and rough Heston is developed during this run. The principal refinement concerns the COS truncation domain and the number of cosine modes \(N\). The preliminary implementation employs a fixed half-width \(L\sqrt{v_0T}\) and a fixed mode count, ignoring correlation, vol-of-vol, and the prevailing variance regime. These choices fail on the stress configurations and are progressively replaced by adaptive rules based on the cumulants of \(\log S_T\), strike-coverage requirements, and additional widening in non-Feller regimes. Writing \(h_n\) and \(N_n\) for the half-width and mode count at iteration \(n\),
	\begin{equation}
		h_n =
		\begin{cases}
			L\sqrt{v_0 T}, & n=1,\\[2pt]
			L\,s_x, & n=2,\\[2pt]
			\max\!\bigl(L\,s_x,\ \max_i|x_i|+0.1\bigr), & 3\le n\le 7,\\[2pt]
			\max\!\bigl(L_{\mathrm{eff}}\,s_x,\ \max_i|x_i|+0.1\bigr), & n\ge 8,
		\end{cases}
	\end{equation}
	\begin{equation}
		N_n =
		\begin{cases}
			128, & 1\le n\le 3,\\[2pt]
			\min\!\bigl(\max(24\,h_n/s_x,\,128),\,1024\bigr), & 4\le n\le 7,\\[2pt]
			\min\!\bigl(\max(24\,h_n/s_x,\,128)\,(2.5-f),\,4096\bigr), & n\ge 8,
		\end{cases}
	\end{equation}
	with \(x_i=\log(K_i/F)\) the selected log-moneynesses and
	\[
	L_{\mathrm{eff}}
	=
	\min\!\bigl(L(2.5-f),16\bigr)
	\]
	for \(f<1\), and \(L_{\mathrm{eff}}=L\) otherwise. For \(f<0.05\), where the analytic estimate \(s_x\) overstates the long-maturity spread, \(s_x\) is capped by
	\[
	\sqrt{\min(|c_2|,\,2.5\,\bar v)},
	\]
	with \(\bar v\) the expected integrated variance of Section~\ref{sec:settings}.
	
	A second refinement concerns deep in- and out-of-the-money contracts. Each strike is priced using the representation with the lighter integrand tail: the call series for \(K\ge F\) and the put series for \(K<F\), with the complementary price recovered using put--call parity. The two estimates are combined through the smooth put--call blending
\begin{equation}
	\widehat{V_c}(K)
	=
	(1-w_{\mathrm{put}})\,V_c(K)
	+
	w_{\mathrm{put}}
	\bigl(V_p(K)+(F-K)e^{-rT}\bigr),
	\qquad
	w_{\mathrm{put}}
	=
	\frac{1}{1+e^{m/\delta}},
\end{equation}
where \(m=\log(K/F)\) and \(\delta=\max(0.25\,s_x,0.01)\). The weight changes smoothly between the call and put representations and assigns equal weight to both at the forward.

The final Heston implementation is a fully adaptive COS method. It satisfies the admissibility checks on the standard and non-Feller stress sets, prices one maturity in approximately \(0.8\) ms, and benchmarks at Level~B. The outstanding deviations occur in the deep-Feller regime, where the variance may approach zero and the resulting tail behaviour is difficult to represent with a finite cosine expansion. This leaves a density deficit and a far out-of-the-money monotonicity residual.

The Bates replication run inherits the adaptive COS rules developed during the Heston run and starts with the jump-augmented characteristic function already implemented. It therefore requires fewer validation iterations. The main additional refinement concerns the truncation rule in the presence of jumps: incorporating the fourth jump cumulant widens the COS truncation domain to account for the heavier tails induced by the jump process. The final Bates implementation prices one maturity in approximately \(1.0\) ms and reaches \(96/100\).

The rough Heston replication run introduces an additional numerical layer because its characteristic function is not available in closed form. The inherited COS machinery can be reused after the fractional-Riccati transform has been computed with sufficient accuracy. The model retains the Heston price dynamics but replaces the square-root variance process by the rough Volterra process
\begin{equation*}
	v_t = v_0 + \frac{1}{\Gamma(\alpha)}\int_0^t (t-s)^{\alpha-1}\bigl[\kappa(\theta-v_s)\,\mathrm{d}s + \sigma_v\sqrt{v_s}\,\mathrm{d}Z_s\bigr],\qquad \alpha = H+\tfrac12,
\end{equation*}
where \(\kappa\) is the mean-reversion speed, \(\theta\) the long-run variance, \(\sigma_v\) the vol-of-variance, and \(H\in(0,\tfrac12]\) the Hurst parameter. For \(\alpha<1\), the variance process exhibits rough behaviour~\cite{el2019rough}. The main task in this run is thus to construct a reliable solver for the associated characteristic function.

The principal refinement is the replacement of the explicit fractional-Adams scheme by an implicit corrector. At each Fourier node, the quadratic equation arising in the fractional update is solved in closed form and the stable root is selected. This stabilizes the transform calculation and allows the time-step selection to be governed primarily by accuracy rather than by stability constraints. The final rough Heston solver prices one maturity in approximately \(27\) ms, benchmarks at Level~C, and reaches \(99/100\). Its remaining residual is a small tail-mass error on the abnormal parameter set, which is diagnosed as a representation limitation rather than an implementation defect.

The three runs use a common repository. The Heston run is performed first and develops most of the adaptive COS machinery. Bates adopts these rules early in its validation, whereas rough Heston can reuse them after its fractional-Riccati solver has become sufficiently reliable. Table~\ref{tab:inheritance} summarizes the inherited and model-specific components.

\begin{table}[H]
	\centering
	\caption{Repository inheritance for the Bates and rough Heston runs: components inherited from the completed Heston validation and model-specific developments.}
	\label{tab:inheritance}
	\footnotesize
	\begin{tabular}{@{}p{0.34\textwidth} p{0.29\textwidth} p{0.29\textwidth}@{}}
		\toprule
		\textbf{Construction (Heston repository)} & \textbf{Bates} & \textbf{Rough Heston} \\
		\midrule
		Cumulant domain, strike-coverage rule, adaptive resolution-\(N\) rule, Feller-aware widening, vectorized series with characteristic-function cache, and out-of-the-money pricing with smooth put--call blending
		& inherited at iteration~2
		& inherited at iteration~4 \\
		\addlinespace
		Characteristic function
		& Heston transform inherited and extended with jumps; validated in the initial implementation (iteration~1)
		& not inherited; fractional-Riccati transform constructed over iterations~1--3 \\
		\addlinespace
		Additional model-specific developments
		& jump fourth-cumulant domain and banded put-side estimator (iterations~3--4)
		& fractional Adams solver, stability-based step-count rule, and implicit corrector (iterations~2--5) \\
		\bottomrule
	\end{tabular}
\end{table}

Across the three runs, adaptive numerical controls consistently replace fixed choices that fail under stressed configurations. The admissibility measurements also reveal deficiencies that benchmark agreement alone does not expose. The per-iteration revision log, including detected regressions and revisions to the validator, along with the per-iteration accuracy, benchmark results, computational costs, and evolution of the dominant stress sets, are provided as Supplementary Material in the GitHub repository: \url{https://github.com/ShQiangLiu/ridge}.
	
\subsection{Validation as a driver of methodological development}\label{sec:meth_dev}

The replication studies of Section~\ref{sec:rep_models} demonstrate that RIDGE can detect and remove implementation defects within an established pricing methodology. We next investigate whether validation can also reveal limitations of the methodology itself and motivate the development of an alternative.

To this end, we consider two models for which the pricing methodology initially generated from the literature proves unsuitable: the Heston--Hull--White (HHW) model~\cite{grzelak2011,HaentjensInTHout2012}, which combines Heston stochastic volatility with a Hull--White short-rate process, and the generalized stochastic-volatility jump-diffusion (G-SVJD) model~\cite{fusari2025}, whose variance diffusion contains the non-affine power-law term $\sigma_vv_t^p$. Neither model admits a directly usable closed-form characteristic function in the classical affine sense, and in both cases the first implementation generated by the language model is based on Monte Carlo simulation.

These validation studies reveal a different failure mode from those encountered in the replication runs. The Monte Carlo implementations can be refined and made internally consistent, but their remaining deficiencies cannot be removed by correcting the code or tuning numerical controls. The underlying limitation is structural: the methodology itself cannot attain the accuracy and efficiency targets prescribed by the validation objective.

The diagnosis consequently shifts from implementation refinement to method selection. Rather than proposing further implementation revisions, it identifies the need for a different pricing representation. In both models, the same structural observation emerges: conditioning on the variance path leaves a conditionally Gaussian log-return, and the Feynman--Kac theorem then reduces pricing to a one-dimensional conditional characteristic-function PDE (CCF-PDE). This formulation eliminates Monte Carlo sampling error and replaces simulation by a deterministic solve.

The validation of both models proceeds through three phases:
\begin{enumerate}
	\item validation of the generated Monte Carlo implementation;
	\item identification of the conditional characteristic-function PDE as a more suitable pricing representation;
	\item validation and refinement of the resulting deterministic solver.
\end{enumerate}

Table~\ref{tab:hhw_progression} records the main revision and category scores at each iteration of the two studies.
	
	\begin{table}[H]
		\centering
		\caption{\revsix{HHW and G-SVJD per-iteration category scores (column headings as in Table~\ref{tab:tier_progression}, with \texttt{[CHC]}$=$[METHOD-CHOICE]). The total is the sum of all six categories. The final HHW implementation reaches $118/120$ at iteration~8 and the final G-SVJD implementation $117/120$ at iteration~7.
    }}
		\label{tab:hhw_progression}
		\footnotesize
		\setlength{\tabcolsep}{2.8pt}
		\begin{tabular}{clcccccc c}
			\toprule
			\textbf{Iteration} & \textbf{Main revision} & \textbf{[BUG]} & \textbf{[ALGO]} & \textbf{[HYP]} & \textbf{[CONS]} & \textbf{[CHC]} & \textbf{[EXP]} & \textbf{Total} \\
			\midrule
			\multicolumn{9}{l}{\textbf{Model: HHW}} \\
			\midrule
			1 & Euler Monte Carlo implementation & 4 & 4 & 4 & 2 & 4 & 20 & 38 \\
			2 & rate initialization; full truncation & 14 & 6 & 5 & 4 & 4 & 20 & 53 \\
			3 & transition to the CCF-PDE & 16 & 10 & 6 & 8 & 20 & 20 & 80 \\
			4 & inherited COS-method refinements & 18 & 14 & 16 & 14 & 20 & 20 & 102 \\
			5 & CIR-moment boundary; adaptive resolution & 18 & 16 & 15 & 13 & 20 & 20 & 102 \\
			6 & phase-only recentering & 18 & 17 & 18 & 15 & 20 & 20 & 108 \\
			7 & out-of-the-money valuation & 18 & 18 & 18 & 16 & 20 & 20 & 110 \\
			8 & efficient deterministic implementation & 20 & 20 & 20 & 20 & 20 & 18 & \textbf{118} \\
			\midrule
			\multicolumn{9}{l}{\textbf{Model: G-SVJD}} \\
			\midrule
			1 & Euler Monte Carlo implementation & 12 & 4 & 8 & 6 & 4 & 20 & 54 \\
			2 & variance-path reconstruction & 16 & 8 & 10 & 8 & 6 & 20 & 68 \\
			3 & transition to the CCF-PDE & 6 & 10 & 8 & 4 & 20 & 20 & 68 \\
			4 & cell averaging; Rannacher smoothing & 12 & 12 & 10 & 6 & 20 & 20 & 80\\
			5 & Lamperti grid; adaptive \(v_{\max}\) & 14 & 14 & 12 & 8 & 20 & 20 & 88 \\
			6 & short-maturity rule; adaptive COS domain & 10 & 16 & 16 & 16 & 20 & 20 & 98 \\
			7 & strike-coverage rule & 20 & 20 & 20 & 20 & 20 & 17 & \textbf{117} \\
			\bottomrule
		\end{tabular}
	\end{table}

\subsubsection{HHW model}\label{sec:meth_hhw}

The HHW model is given by
\begin{align*}
	\mathrm{d}S_t/S_t &= r_t\,\mathrm{d}t + \sqrt{v_t}\,\mathrm{d}W_1(t),\\
	\mathrm{d}v_t &= \kappa(\eta - v_t)\,\mathrm{d}t + \sigma_1\sqrt{v_t}\,\mathrm{d}W_2(t),\\
	\mathrm{d}r_t &= a(b - r_t)\,\mathrm{d}t + \sigma_2\,\mathrm{d}W_3(t),
\end{align*}
with correlations $\mathrm{d}[W_1,W_2]_t = \rho_{12}\,\mathrm{d}t$ and $\mathrm{d}[W_1,W_3]_t = \rho_{13}\,\mathrm{d}t$, while the short rate and variance are uncorrelated, where we require $\rho_{23}=0$ and $\rho_{12}^2 + \rho_{13}^2 \le 1$. The main obstacle is that the non-zero equity--rate correlation $\rho_{13}$ breaks the classical affine structure, so the model admits no closed-form characteristic function through the standard Riccati framework.

\paragraph{Phase 1: validating the generated Monte Carlo implementation.}

The generated implementation is a basic Euler Monte Carlo simulation, and the first two iterations correct several implementation defects. The short-rate drift is initialized independently of the market rate and the variance process is  partially truncated, both of which become apparent under degeneration to the Heston limit. The consistency checks and degeneration benchmark identify both deficiencies. Setting the initial short rate to $r_0=r$ and replacing the variance discretization by full truncation remove these defects, after which the simulation prices correctly to within its sampling noise.

The remaining deficiencies exhibit a different pattern. Further revisions improve neither the benchmark grade nor the computational efficiency in a meaningful way. The primary limitation no longer lies in the implementation but in the methodology itself. At the inverse-square-root convergence rate of Monte Carlo, reducing the implied-volatility deviation to the target level of $10^{-5}$ would require on the order of $10^{10}$ paths, which is incompatible with the efficient-pricing objective of the validation framework. The method, rather than its realization, has become the main limitation.

\paragraph{Phase 2: identifying a different methodology.}

A \texttt{[METHOD-\allowbreak CHOICE]} diagnosis is raised once the residual pattern can no longer be removed by correcting the code or tuning numerical controls within the Monte Carlo framework. The focus of the revision changes accordingly, from further refinement of the simulation to a pricing representation without sampling error.

As outlined above, the diagnosis gives rise to a transition away from Monte Carlo, after which variance-path conditioning and the Feynman–Kac theorem yield a conditional characteristic-function PDE representation. For HHW, this formulation emerges in the third iteration as the conditional characteristic-function PDE of Proposition~\ref{prop:hhw}, derived in Appendix~\ref{app:hhw}.

The transform is obtained by conditioning on the variance path. Because the variance and short-rate Brownian motions are uncorrelated, the log-forward is conditionally Gaussian given that path, and a linear change of variable in the variance makes the reduction exact. The conditional transform of the variance process satisfies, by the Feynman--Kac theorem, a one-dimensional backward parabolic equation. Solving this equation once per maturity yields the exact characteristic function, after which the COS valuation is applied unchanged. The \texttt{[METHOD-CHOICE]} score consequently rises to its maximum value.

\paragraph{Phase 3: validating and refining the new methodology.}

With the methodology established, the next iterations refine the implementation of the deterministic solver. Iteration~4 applies the inherited COS-method refinements in a single revision, bringing the default parity error to machine precision and reducing the test-suite violation count from $644$ to $119$. The total score increases from \(80\) to \(102\).

The subsequent iterations introduce the numerical machinery required for a transform obtained from a PDE solve. Each iteration addresses a deficiency exposed by the preceding one, and Table~\ref{tab:hhw_thirdphase} summarizes the resulting refinements and their quantitative effect.

Because the terminal data oscillates in the variance according to $\exp(\mathrm{i}u\rho_{12}v/\sigma_1)$, the variance grid must be wide enough to contain the relevant probability mass while being sufficiently fine to resolve the oscillation. Iteration~5 consequently places the far boundary several standard deviations above the variance mean,
\[
v_{\max}=m(T)+n_\sigma s(T),
\]
where $m(T)$ and $s(T)$ denote the mean and standard deviation of $v_T$. The boundary is closed by a Neumann condition consistent with the terminal value and the mesh width is linked to the oscillation frequency through the resolution condition
\[
\Delta v\lesssim \frac{\sigma_1}{|\rho_{12}|u_{\max}},
\]
where $u_{\max}$ is the largest Fourier argument appearing in the COS expansion. These modifications improve the degeneration benchmark from a failing grade to Level~B and reduce the convexity residual by more than two orders of magnitude.

For small vol-of-variance, the complex source term $c_{\mathrm{HHW}}$ of~\eqref{eq:hhw_pde_prop} develops a rapidly varying phase. Iteration~6 recenters  the imaginary phase along the variance mean path,
\[
m(t)=\mathbb{E}[v_t],
\]
thereby removing the divergence. The benchmark consequently attains Level~C.

Iteration~7 removes the kink introduced by the smooth put--call blending. The implementation instead prices the out-of-the-money option directly from the COS series and recovers the complementary price by the put--call parity, reducing the monotonicity and convexity residuals to nearly zero. Finally, iteration~8 addresses computational efficiency. The pricing problem is reduced to a single Crank--Nicolson solve per parameter set and maturity, batched across all COS modes in one Thomas sweep, lowering the per-maturity runtime from $548$ to $471$\,ms.

The final implementation reaches \(118/120\), with parity, monotonicity, and convexity violations reduced to machine precision and the benchmark reaching Level~C, corresponding to a maximum at-the-money implied-volatility deviation of $5.07\times10^{-5}$ across the maturity grid. The sole remaining \texttt{[EXPECTED]} residual is the pre-registered density-integral deviation under Feller violation ($\sim2\times10^{-3}$), which is retained as a representation limitation rather than an implementation defect.
	
	\begin{table}[H]
		\centering
		\revten{
			\caption{HHW third phase (iterations 5--8): principal refinements and their measured effect.}
			\label{tab:hhw_thirdphase}
			\small
			\renewcommand{\arraystretch}{1.6}
			\begin{tabular}{@{}c p{6.7cm} p{4.6cm}@{}}
				\toprule
				\textbf{Iter.} & \textbf{Principal refinement} & \textbf{Measured effect} \\
				\midrule
				5 & CIR-moment variance boundary; far-field Neumann condition; variance-resolution rule & benchmark $1.4\times10^{-2}$ (F)$\,\to\,2.1\times10^{-4}$ (B); convexity $L_3\;1.0\times10^{-3}\to1.1\times10^{-5}$ \\
				6 & phase-only recentering of the imaginary source along the variance mean path; dissipative real part preserved & benchmark $2.1\times10^{-4}$ (B)$\,\to\,5.7\times10^{-5}$ (C) \\
				7 & out-of-the-money valuation with put--call parity recovery, replacing the smooth put--call blending & monotonicity $L_2\;4.1\times10^{-5}\to2.1\times10^{-6}$; convexity $L_3\to0$ (blend kink removed) \\
				8 & one Crank--Nicolson solve per parameter set and maturity, batched across all COS modes in one Thomas sweep & per-maturity time $548\to471$\,ms; test-suite violations $\to16$; total $110\to118/120$ \\
				\bottomrule
		\end{tabular}}
	\end{table}
	
\subsubsection{G-SVJD model}\label{sec:meth_gsvjd}

The G-SVJD model \cite{fusari2025} is given by
\begin{equation*}
	\frac{\mathrm{d}S_t}{S_t}
	=
	(r-q-\lambda m_J)\,\mathrm{d}t
	+
	\sqrt{v_t}\,\mathrm{d}W_t
	+
	(J_t - 1)\,\mathrm{d}N_t,
	\qquad
	\mathrm{d}v_t
	=
	\kappa(\bar{\nu}-v_t)\,\mathrm{d}t
	+
	\sigma_v v_t^p\,\mathrm{d}Z_t,
\end{equation*}
with correlation $\mathrm{d}[W,Z]_t=\rho\,\mathrm{d}t$, parameter $p > 0$, log-normal jumps $\log J_j\sim\mathcal{N}(\mu_y,\sigma_y^2)$, and compensator $m_J=e^{\mu_y+\sigma_y^2/2}-1$. The main numerical challenge is the non-affine variance diffusion $\sigma_v v_t^p$.

\paragraph{Phase 1: exposing the structural weakness.}

The generated implementation is a full-Euler Monte Carlo simulation. The first iteration establishes the baseline Euler scheme, whose direct discretization of both the stock and variance paths introduces the usual bias and fails several consistency checks. Iteration~2 replaces the return discretization by the conditional-Gaussian variance-path integration derived in Appendix~\ref{app:gsvjd}. This makes put--call parity  to machine precision and removes the return-discretization error, so that the consistency hierarchy passes.

The remaining limitation is twofold and structural. Simulating the variance path still carries the full-truncation Euler bias associated with a variance process that can reach zero, producing an $O(\sqrt{\Delta t})$ error plateau that variance-reduction techniques cannot remove. Moreover, the Monte Carlo error decreases at the inverse-square-root rate, making accurate pricing computationally expensive. The degeneration benchmark consequently stalls at Level~A.

\paragraph{Phase 2: identifying a better method.}

The persistent error plateau gives rise to a \texttt{[METHOD-\allowbreak CHOICE]} diagnosis: the residual can no longer be removed within the Monte Carlo methodology, and the revision target changes from refining the Monte Carlo simulation to finding a pricing representation without sampling error.

For G-SVJD, the alternative representation additionally requires a change of variable.  In the third iteration, the conditional characteristic-function PDE of Proposition~\ref{prop:gsvjd} is obtained; its derivation is given in Appendix~\ref{app:gsvjd}. Conditioning on the variance path again leaves a conditionally Gaussian log-return, but the non-affine variance diffusion introduces an unstable imaginary advection term in the reduced equation. A change of variable absorbs the correlation coupling into the terminal value and restores a real, stable operator. The Feynman--Kac theorem then reduces the conditional transform to a one-dimensional backward equation. Monte Carlo sampling error and variance-path discretization bias are thereby replaced by the numerical error of a deterministic PDE solve. The \texttt{[METHOD-CHOICE]} score consequently rises to its maximum value.

\paragraph{Phase 3: guiding the refinement.}

With the methodology established, the subsequent iterations refine the deterministic solver. The new representation introduces an additional numerical complication: for $p>\tfrac12$, the change of variable that stabilizes the equation produces a drift term proportional to $v^{1/2-p}$, which diverges as the variance approaches zero. A naive discretization samples a singular drift and produces a density whose integral departs from unity.

The subsequent iterations address this defect and the remaining COS-method errors. Table~\ref{tab:gsvjd_thirdphase} summarizes the refinements and their measured effects. Iteration~4 replaces the singular drift by its average over the first variance cell and damps the nonsmooth terminal data through implicit-Euler (Rannacher) start-up steps, reducing the density-integral deviation from $1.56$ to $9.7\times10^{-2}$. Iteration~5 places the variance nodes uniformly in the Lamperti coordinate
\[
Y=\frac{v^{1-p}}{\sigma_v(1-p)},
\]
which transforms the power-law diffusion $\sigma_vv^p$ to a constant coefficient and clusters nodes near $v=0$, where the original diffusion degenerates. The variance domain is also sized from the moments of $v_T$, reducing the density-integral deviation to $1.0\times10^{-3}$ and bringing the degeneration benchmark below $10^{-4}$ at the longer maturities. Iteration~6 introduces a short-maturity step-count floor together with a jump-aware cosine half-width,
\[
s_x=\sqrt{|c_2|+\sqrt{c_4^{J}}},
\]
where $c_2$ and $c_4^{J}$ denote the diffusive and jump cumulants. This improves the benchmark from a failing grade to Level~C. Finally, iteration~7 introduces the strike-coverage rule and eliminates the short-maturity monotonicity artifact.

The final implementation reaches \(117/120\). Its sole residual concerns the jump-dominated density tails, which cannot be fully represented on a finite COS truncation domain and are retained as a representation limitation rather than an implementation defect.

\begin{table}[H]
	\centering
	\revten{
		\caption{G-SVJD third phase (iterations 4--7): principal refinements and their measured effects, in the format of Table~\ref{tab:hhw_thirdphase}.}
		\label{tab:gsvjd_thirdphase}
		\small
		\renewcommand{\arraystretch}{1.5}
		\begin{tabular}{@{}c p{6.7cm} p{4.6cm}@{}}
			\toprule
			\textbf{Iter.} & \textbf{Principal refinement} & \textbf{Measured effect} \\
			\midrule
			4 & cell average of the singular drift over the first variance cell; implicit-Euler (Rannacher) start-up for the nonsmooth terminal data & density integral $1.56\to9.7\times10^{-2}$; price $11.29\to9.64$ (reference $9.43$) \\
			5 & Lamperti-grid discretization, clustering nodes near $v=0$; variance domain sized from the moments of $v_T$ & density integral $9.7\times10^{-2}\to1.0\times10^{-3}$; degeneration benchmark below $10^{-4}$ for $T\ge\tfrac12$ \\
			6 & short-maturity step-count floor; jump-aware COS half-width & benchmark (F)$\,\to\,$Level~C ($6.3\times10^{-5}$) \\
			7 & strike-coverage rule extending the cosine interval to all requested log-strikes at every maturity & monotonicity $L_2\;21.4\to0$; total score \(117/120\) \\
			\bottomrule
		\end{tabular}}
\end{table}

HHW is validated first and develops the conditional characteristic-function PDE definition and its numerical solution rules before G-SVJD is considered. When G-SVJD undergoes the same methodological transition, these components are already available in the repository. The subsequent iterations can then focus on the additional numerical challenges introduced by the G-SVJD change of variable. Table~\ref{tab:gsvjd_inheritance} summarizes the resulting inheritance structure.

\begin{table}[H]
	\centering
	\caption{Repository components inherited by the G-SVJD run, together with their origin and the iteration in which they are first used.}
	\label{tab:gsvjd_inheritance}
	\footnotesize
	\begin{tabular}{@{}p{0.52\textwidth} p{0.40\textwidth}@{}}
		\toprule
		\textbf{Inherited component (source)} & \textbf{Role in G-SVJD (iteration)} \\
		\midrule
		Conditional characteristic-function PDE: variance-path conditioning and Feynman--Kac reduction to a one-dimensional PDE (HHW) & method transition (iteration~3) \\
		\addlinespace
		COS-method bundle: cumulant truncation, strike-coverage rule, resolution-\(N\) rule, and out-of-the-money pricing with smooth put--call blending (Heston) & COS valuation (iteration~3) \\
		\addlinespace
		Variance-domain sizing: far boundary determined from the moments of the variance process (HHW) & variance-grid set-up (iteration~5) \\
		\addlinespace
		Batched solver: one Crank--Nicolson sweep across all COS modes (HHW) & batched PDE solve (iteration~3) \\
		\addlinespace
		Jump-aware cosine half-width: domain widening through the jump fourth cumulant (Bates) & COS-domain set-up (iteration~6) \\
		\bottomrule
	\end{tabular}
\end{table}

The two methodology-development studies show that validation can identify limitations of the pricing methodology, rather than only defects in its implementation. In both cases, conditioning on the variance path followed by a Feynman--Kac reduction leads to a deterministic pricing framework that avoids Monte Carlo sampling error and variance-path discretization bias. The inheritance chain from Heston through HHW to G-SVJD allows the later studies to reuse previously developed numerical components and concentrate on model-specific difficulties. The  per-iteration revision logs for both models, along with the accuracy, benchmark results, computational costs, and per-iteration stress breakdowns underlying the diagnostic scores, are provided as Supplementary Material in the GitHub repository.
	
	\FloatBarrier
\subsection{Performance of the new methods}\label{sec:performance}

Figure~\ref{fig:speed_accuracy} compares, for both models at $T=0.5$, the mean implied-volatility error across the central strike grid against the per-maturity CPU time. Each implementation traces an accuracy--cost frontier obtained by varying a single numerical parameter: the Monte Carlo implementation by its number of simulated paths, the final CCF-PDE solver by its grid resolution, and the iteration-4 implementation by the variance-domain cap that limits its attainable accuracy. The reference is a converged CCF-PDE solution computed on a fine grid ($N_V=240$, $N_t=120$ for HHW; $N_V=200$ with $60$ time steps per year for G-SVJD), with the resolution increased until the prices  no longer change significantly. These reference values are computed with the same numerical method and serve to measure finite-resolution errors and construct the accuracy–cost frontiers.

\begin{figure}[t]
	\centering
	\includegraphics[width=0.49\linewidth]{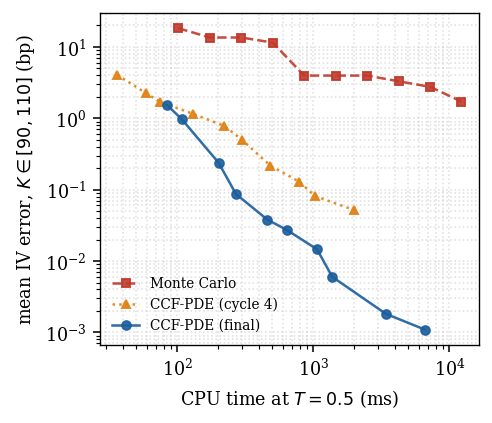}\hfill
	\includegraphics[width=0.49\linewidth]{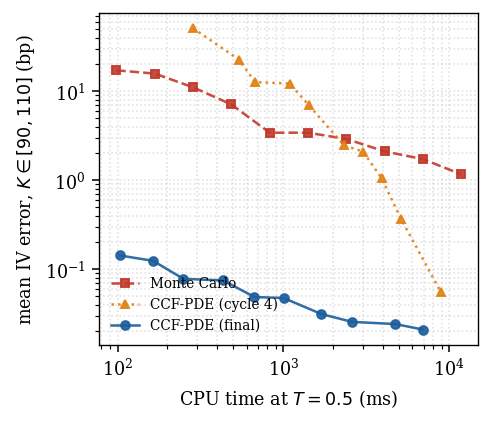}
	\caption{\revten{Speed and accuracy at $T=0.5$ for HHW (left) and G-SVJD (right): mean implied-volatility error across the central strike grid $K\in[90,110]$ against per-maturity CPU time, timed single-threaded on one virtualized Intel Xeon core at $2.80$\,GHz}. Each method is represented by its accuracy--cost frontier obtained from a common geometric sequence of refinements, with the converged fine-grid CCF-PDE solution taken as reference. Monte Carlo is varied over its antithetic path count at a fixed Euler time step, the final CCF-PDE solver over its grid resolution, and the iteration-4 implementation over its variance-domain cap. Both panels use the base configuration $S_0=100$, $v_0=0.04$, $q=0$: for HHW $\kappa=2$, $\eta=0.04$, $\sigma_1=0.3$, $\rho_{12}=-0.7$, $\rho_{13}=0.3$, $a=0.5$, $\sigma_2=0.02$, with flat initial rate $r_0=0.02$; for G-SVJD $\kappa=2$, $\bar\nu=0.04$, $\sigma_v=0.3$, $p=0.7$, $\rho=-0.7$, $\lambda=0.1$, $\mu_y=-0.10$, $\sigma_y=0.15$, with $r=0.02$.}
	\label{fig:speed_accuracy}
\end{figure}

Figure~\ref{fig:speed_accuracy} shows the performance gains resulting from the methodological developments of Section~\ref{sec:meth_dev}. For both HHW and G-SVJD, the final CCF-PDE formulation achieves sub-basis-point implied-volatility errors in a few hundred milliseconds and delivers errors between one and three orders of magnitude smaller than those of the Monte Carlo implementation at comparable computational cost.

The Monte Carlo curve initially improves at the expected inverse-square-root rate as the number of simulated paths increases. Eventually, however, both models reach an error plateau. The remaining error is no longer statistical but is dominated by the time-discretization bias associated with simulating a variance process that can approach zero. Additional paths therefore reduce the variance of the estimator but have little effect on the total pricing error.
The iteration-4 implementation occupies an intermediate position between Monte Carlo and the final CCF-PDE solver and illustrates the role of validation-guided refinement. Immediately after the transition to the conditional characteristic-function PDE, the pricing methodology is already more accurate than Monte Carlo, but its implementation still contains numerical deficiencies that prevent it from attaining its full accuracy.

For HHW, the dominant limitation is an overly restrictive truncation of the variance domain. The iteration-4 implementation employs the fixed cap
\[
v_{\max}=3\max(v_0,\eta),
\]
which excludes a non-negligible portion of the variance distribution and leaves an accuracy floor of approximately four basis points. Refining the variance or strike grids cannot remove this error, because the main source of inaccuracy is the placement of the boundary. Iteration~5 resolves this deficiency by sizing the variance domain from the moments of the variance process, producing a substantially more accurate solution at essentially unchanged computational cost.

For G-SVJD, the dominant limitation is instead the singular behaviour of the transformed drift near $v=0$. The Lamperti coordinate introduced in iteration~5 clusters grid points in this region, while the variance moments again determine an appropriate far boundary. Once these modifications are introduced, the residual error floor disappears and the accuracy--cost curve moves down to the level of the final solver.

The two studies show that replacing an inadequate pricing methodology can improve the accuracy--cost trade-off by several orders of magnitude. They also show that the methodological change is the first step: the new deterministic framework still requires numerical validation and refinement. The final curves in Figure~\ref{fig:speed_accuracy} reflect the combined effect of method selection and subsequent implementation refinement.
	
\section{Ablation and reproducibility}\label{sec:ablation}
This section presents the baseline comparison and reproducibility experiments that isolate the contribution of the validation operator.  We investigate two further questions: how much of the observed refinement is attributable to the validation operator, and does RIDGE require the implementation generator and diagnosis operator to use the same language model?

We design three experimental settings. The experiments of Section~\ref{sec:experiments}, where the pricing implementation and performs the diagnosis within RIDGE share the same language model, are referred to below as Experiment~2. Experiment~1 removes the validation operator. The language model generates a Heston~\cite{heston1993} pricing implementation from the model specification and the COS method~\cite{fang2008}, and revises it iteratively using code execution and comparison with a single published reference value. No diagnosis stage, consistency checks, or stress testing is available.  Experiment~3 retains the validation operator but uses an external language model as the implementation generator. At each iteration, this model receives the current blueprint and diagnostic report and generates the revised blueprint and pricing implementation, while Claude Opus~4.8 performs the validation, diagnosis, and repository updates. Two external generators are considered in the Experiment~3 runs: DeepSeek (deepseek-v4-pro) and ChatGPT (gpt-5.5).

Experiment 1 represents a standard LLM coding agent as the baseline, while Experiment 2 adds the RIDGE-specific validation layer. Furthermore, Experiment~3 tests whether the same validation procedure remains effective when implementation generation and diagnosis are assigned to different language models.
Table~\ref{tab:ablation_states} records, for each Heston run, the constructions present in the final implementation and the version (Experiment~1) or iteration (Experiments~2 and~3) in which each first appears. The upper block contains a mathematical core shared by all three runs. The lower block contains regime-adaptive refinements that relate the truncation interval and mode count to the selected strikes, maturity, and Feller ratio $f$.

\begin{table}[H]
	\centering
	\caption{Numerical procedures appearing in the three Heston runs and the version (Experiment~1) or iteration (Experiments~2 and~3) in which each first appears. A dash indicates that a procedure is not introduced; a range indicates development over consecutive iterations. The upper block lists procedures shared by all runs and the lower block those developed in the operator-guided runs.}
	\label{tab:ablation_states}
	\footnotesize
	\setlength{\tabcolsep}{6pt}
	\begin{tabular}{lccc}
		\toprule
		\textbf{Construction} & \textbf{Exp.\ 1 version} & \textbf{Exp.\ 2 iteration} & \textbf{Exp.\ 3 iteration} \\
		\midrule
		\multicolumn{4}{l}{\textit{Shared by all runs}} \\
		Analytic characteristic function and COS valuation & 1 & 1 & 1 \\
		Cumulant-based truncation domain & 1 & 2 & 1 \\
		General-endpoint payoff coefficients & 1 & 2 & 1 \\
		Put pricing via put--call parity & 2 & 3 & 4 \\
		Strike vectorization & 3 & 4 & 1 \\
		\midrule
		\multicolumn{4}{l}{\textit{Developed only in the operator-guided runs}} \\
		Strike-coverage rule & -- & 3 & 3 \\
		Resolution-scaled mode count & -- & 4 & 3 \\
		Out-of-the-money pricing & -- & 6--7 & 4 \\
		Feller-aware domain widening & -- & 8--9 & 5 \\
		Feller floor for the mode count & -- & 8--9 & 5 \\
		Smooth put--call blending & -- & 10 & -- \\
		\bottomrule
	\end{tabular}
\end{table}

The difference between the three runs is pronounced. Experiment~1 stops after three versions with only the shared mathematical core. None of the adaptive frameworks is introduced. Comparison with a single reference value provides no evidence that the fixed truncation interval and mode count fail outside the default configuration. The deficiencies exposed by the stress configurations in Experiment~2 thus remain undetected.

Both operator-guided runs, Experiments~2 and~3, develop the adaptive layer. This remains the case when the implementation generator is external. In the Heston run of Experiment~3, the validation operator identifies a fatal implementation error in the first iteration and subsequently detects overflow caused by heavy tails on the high-volatility and violated-Feller stress sets. The resulting diagnoses lead to adaptive resolution, strike coverage, out-of-the-money pricing, and Feller-aware truncation rules. After these revisions, the monotonicity and convexity checks are satisfied across all stress sets, and the final implementation reaches the \textit{halt} state.
	
The adaptive layer emerges as a consequence of the validation operator. It is developed both when generation and diagnosis use the same language model and when the implementation generator is external, but is absent when the validation operator is removed.

Table~\ref{tab:ablation_engines} compares the three final implementations on the default configuration and on representative extremes of the stress parameters, using the common measurement protocol of Section~\ref{sec:iteration}. The comparison evaluates whether the adaptive settings identified in Table~\ref{tab:ablation_states} translate into improved admissibility and pricing accuracy outside the configuration used by the unguided run.

\begin{table}[H]
	\centering
	\caption{The three final Heston implementations---the unguided run (Exp.\ 1), the full RIDGE run (Exp.\ 2), and the external-generator run (Exp.\ 3)---evaluated on the default configuration and a representative extreme of each stress parameter. $L_2$ and $L_3$ are the percentages of monotonicity and convexity checks that fail; $L_4$ is the mean density-integral deviation. The error columns report the mean relative price error against the reference over $|m|\le0.25$ and $|m|\le0.5$, with $m=\log(K/F)$. Time is measured per maturity, and the abnormal set is $\kappa=0.2$, $\theta=0.004$, $\sigma_v=0.9$.}
	\label{tab:ablation_engines}
	\scriptsize
	\setlength{\tabcolsep}{3pt}
	\begin{tabular}{cclccc cc r}
		\toprule
		\textbf{Set} & \textbf{$f$} & \textbf{Implementation} & \textbf{$L_2$ (\%)} & \textbf{$L_3$ (\%)} & \textbf{$L_4$} & \textbf{err $|m|\!\le\!0.25$} & \textbf{err $|m|\!\le\!0.5$} & \textbf{Time, ms} \\
		\midrule
		\multirow{3}{*}{Default} & \multirow{3}{*}{$1.78$} & Exp.\ 1 & $8.8$ & $2$ & $8.8\!\times\!10^{-5}$ & $5.7\!\times\!10^{-2}$ & $0.106$ & 0.3 \\
		&  & Exp.\ 2 & $0$ & $0$ & $8.8\!\times\!10^{-5}$ & $5.0\!\times\!10^{-7}$ & $5.9\!\times\!10^{-7}$ & 0.6 \\
		&  & Exp.\ 3 & $0$ & $0$ & $8.0\!\times\!10^{-5}$ & $9.5\!\times\!10^{-7}$ & $5.3\!\times\!10^{-6}$ & 1.8 \\
		\midrule
		\multirow{3}{*}{$\kappa=0.1$} & \multirow{3}{*}{$0.0889$} & Exp.\ 1 & $9.2$ & $1$ & $3.3\!\times\!10^{-3}$ & $5.1\!\times\!10^{-2}$ & $0.107$ & 0.2 \\
		&  & Exp.\ 2 & $0$ & $0$ & $2.7\!\times\!10^{-3}$ & $3.8\!\times\!10^{-4}$ & $1.6\!\times\!10^{-3}$ & 1.4 \\
		&  & Exp.\ 3 & $0$ & $0$ & $3.1\!\times\!10^{-3}$ & $1.4\!\times\!10^{-6}$ & $6.3\!\times\!10^{-6}$ & 1.2 \\
		\midrule
		\multirow{3}{*}{$\theta=0.01$} & \multirow{3}{*}{$0.444$} & Exp.\ 1 & $9.3$ & $2.1$ & $3.6\!\times\!10^{-4}$ & $5.9\!\times\!10^{-2}$ & $0.112$ & 0.3 \\
		&  & Exp.\ 2 & $0$ & $0$ & $3.6\!\times\!10^{-4}$ & $1.9\!\times\!10^{-10}$ & $4.2\!\times\!10^{-9}$ & 1.3 \\
		&  & Exp.\ 3 & $0$ & $0$ & $3.6\!\times\!10^{-4}$ & $1.5\!\times\!10^{-6}$ & $2.5\!\times\!10^{-6}$ & 1.8 \\
		\midrule
		\multirow{3}{*}{$\sigma_v=1$} & \multirow{3}{*}{$0.160$} & Exp.\ 1 & $9.4$ & $2.1$ & $4.3\!\times\!10^{-3}$ & $6.0\!\times\!10^{-2}$ & $0.108$ & 0.3 \\
		&  & Exp.\ 2 & $0$ & $0$ & $4.3\!\times\!10^{-3}$ & $7.0\!\times\!10^{-7}$ & $1.7\!\times\!10^{-6}$ & 1.4 \\
		&  & Exp.\ 3 & $0$ & $0$ & $4.3\!\times\!10^{-3}$ & $7.0\!\times\!10^{-5}$ & $1.2\!\times\!10^{-4}$ & 1.3 \\
		\midrule
		\multirow{3}{*}{$\rho=-0.95$} & \multirow{3}{*}{$1.78$} & Exp.\ 1 & $9.2$ & $2.1$ & $1.8\!\times\!10^{-4}$ & $5.9\!\times\!10^{-2}$ & $0.11$ & 0.3 \\
		&  & Exp.\ 2 & $0$ & $0$ & $1.8\!\times\!10^{-4}$ & $1.8\!\times\!10^{-6}$ & $2.3\!\times\!10^{-6}$ & 0.6 \\
		&  & Exp.\ 3 & $0$ & $0$ & $8.9\!\times\!10^{-5}$ & $3.0\!\times\!10^{-3}$ & $1.9\!\times\!10^{-3}$ & 0.9 \\
		\midrule
		\multirow{3}{*}{$v_0=0.01$} & \multirow{3}{*}{$1.78$} & Exp.\ 1 & $13.2$ & $1.2$ & $1.4\!\times\!10^{-4}$ & $0.124$ & $0.18$ & 0.3 \\
		&  & Exp.\ 2 & $0$ & $0$ & $1.4\!\times\!10^{-4}$ & $1.9\!\times\!10^{-6}$ & $1.7\!\times\!10^{-6}$ & 0.8 \\
		&  & Exp.\ 3 & $0$ & $0$ & $1.2\!\times\!10^{-4}$ & $3.7\!\times\!10^{-6}$ & $5.2\!\times\!10^{-6}$ & 1.6 \\
		\midrule
		\multirow{3}{*}{Abnormal} & \multirow{3}{*}{$0.00198$} & Exp.\ 1 & $14.1$ & $11.1$ & $1.33$ & $0.127$ & $0.178$ & 0.3 \\
		&  & Exp.\ 2 & $3.6$ & $0$ & $3.7\!\times\!10^{-2}$ & $7.4\!\times\!10^{-3}$ & $1.5\!\times\!10^{-2}$ & 1.5 \\
		&  & Exp.\ 3 & $0$ & $0$ & $1.1\!\times\!10^{-2}$ & $1.5\!\times\!10^{-4}$ & $1.6\!\times\!10^{-4}$ & 1.3 \\
		\bottomrule
	\end{tabular}
\end{table}

The unguided implementation differs markedly from the two operator-guided implementations. It violates strike monotonicity on approximately nine to fourteen percent of adjacent-strike pairs and develops additional convexity violations on the abnormal set. Its mean relative pricing error is already several percent over $|m|\le0.25$ and increases to between ten and twenty percent over the wider range $|m|\le0.5$, which extends beyond the effective coverage of its fixed truncation domain.

The two operator-guided implementations, by contrast, exhibit no monotonicity or convexity violations on the default and individual stress configurations. Their relative pricing errors generally range from a few parts in $10^{-6}$ to $10^{-3}$, while a refinement study identifies the small density-integral residuals as a discretization floor rather than an implementation deficiency. The remaining deviations are concentrated on the pre-registered abnormal configuration at long maturity. There, the heavy-tailed density and extreme out-of-the-money strikes cannot be represented exactly on a finite COS truncation domain. The corresponding density-integral deviations are $3.7\times10^{-2}$ for Experiment~2 and $1.1\times10^{-2}$ for Experiment~3 and are classified as \texttt{[EXPECTED]}.
	
The comparison of the three experiments identifies the contribution of the validation operator. Refinement against a single benchmark configuration, as in Experiment~1, can produce an implementation that prices that configuration accurately and appears locally consistent, yet still exhibits substantial pricing errors and arbitrage violations when evaluated in other parameter regimes and over wider strike ranges. The broader evidence supplied by the validation operator leads to the regime-adaptive settings and robustness observed in Experiments~2 and~3.

Beyond the Heston ablation study, Experiment~3, the external-generator reproducibility experiment, was extended to three additional model classes and two external language models.

DeepSeek (deepseek-v4-pro) served as the generator for the Bates and rough Heston studies. Bates reached a validated implementation in five iterations, matching the corresponding Experiment~2 replication run, with a degeneration-to-Heston benchmark of $7.8\times10^{-9}$ and a runtime of $0.35$\,ms per maturity. Rough Heston reached a validated implementation in nine iterations, compared with six iterations in Experiment~2. The additional revisions were devoted mainly to performance vectorization and coverage refinement. Its final implementation attains a degeneration benchmark of $5.6\times10^{-5}$ at $114$\,ms per maturity, with accuracy comparable to Experiment~2 at a moderately higher computational cost.

For G-SVJD, ChatGPT (gpt-5.5) served as the generator. Unlike the HHW study, in which the external generator independently reproduced the conditional characteristic-function PDE of Experiment~2, the G-SVJD study arrived at a different pricing methodology: a hybrid Fourier finite-difference solver that transforms the log-price dimension and solves a one-dimensional complex-valued PDE for each Fourier frequency. The final implementation reaches a validated state in four iterations and reproduces the Heston degeneration ($p=\tfrac12$, $\lambda=0$) to within $10^{-5}$, with a runtime of approximately $360$\,ms per maturity on the non-abnormal configurations. Its higher computational cost reflects the convergence properties of the alternative methodology rather than an implementation deficiency.

These additional studies extend the Heston ablation and reproducibility results beyond a single model and generator. With DeepSeek and ChatGPT as external generators, RIDGE reaches validated implementations for affine jump-diffusion, rough-volatility, and non-affine models. The external generator may reproduce the methodology obtained in Experiment~2 or arrive at a different numerical techniques, as in the G-SVJD study. In either case, the validation operator evaluates the implementation through the same deterministic measurements, identifies the remaining deficiencies, and directs subsequent revisions. The results suggest that the validation process, rather than the choice of implementation generator, is the main source of the observed regime adaptivity and final validation quality.
	
\section{Discussion and conclusion}\label{sec:discussion}

\subsection{Discussion}

Across the five models, the main failures are not programming mistakes and rarely involve mathematical errors. Two broader numerical deficiencies recur.

The first is insufficient adaptivity. Numerical controls that appear adequate for a default configuration often fail across maturities, moneyness levels, and stressed parameter regimes. A fixed step count for the rough-Heston Volterra solver is no more reliable than a fixed mode count for the cosine expansion. In each study, adaptive rules emerge in response to measured violations rather than being prescribed in advance. Once incorporated into the repository, several of these rules transfer to subsequent model classes and provide a stronger starting point for later validation studies.

The second deficiency is a mismatch between the numerical methodology and the structure of the model. This occurs in the HHW and G-SVJD studies. Their generated Monte Carlo implementations can be made internally consistent, but their convergence properties remain incompatible with the accuracy and efficiency objectives of the validation process. In these cases, further implementation refinement is insufficient and the methodology itself becomes the object of diagnosis.

The HHW and G-SVJD studies illustrate a broader role for validation in numerical method development. In both cases, the persistent error pattern leads to a \texttt{[METHOD-\allowbreak CHOICE]} diagnosis and a transition to a deterministic pricing framework. Conditioning on the variance path renders the log-return conditionally Gaussian; an appropriate transformation followed by the Feynman--Kac theorem then reduces valuation to a one-dimensional backward PDE in the variance variable; see Appendices~\ref{app:hhw} and~\ref{app:gsvjd}. The resulting techniques avoid Monte Carlo sampling error and the discretization bias associated with direct simulation of a variance process near zero, without introducing an affine approximation of the original model. More generally, these studies show how persistent validation failures can provide evidence that the numerical representation, rather than its implementation, should be reconsidered.

A further observation concerns benchmarking. No single reference mechanism is suitable across all model classes. Published numerical values provide the strongest direct evidence when available, while degeneration limits are particularly informative when a complex model reduces to a well-validated simpler one. For models lacking either form of reference, local affine approximations and consistency under numerical refinement provide weaker but still useful evidence. The benchmark hierarchy must therefore adapt to the analytical structure of the model and to the reference information available.

Stress testing is equally important. Several substantial deficiencies remain invisible on the default configuration and appear only under deep Feller violation, far-from-stationary variance regimes, or extreme jump parameters. Combined stresses can reveal failures that are absent along every individual stress direction. These measurements motivate several of the adaptive settings retained in the repository, including Feller-aware domain widening and jump-aware truncation rules. Residual deviations that remain after validation occur on pre-registered abnormal configurations and are reported under \texttt{[EXPECTED]}, together with their numerical or theoretical explanation, rather than being removed through further parameter tuning.

Taken together, the five case studies suggest that autonomous validation has a broader role than detecting coding errors. It can expose insufficient numerical adaptation, identify a mismatch between a model and its pricing methodology, guide the development and refinement of alternative numerical methods, and accumulate validation experience for subsequent model classes.
	
\subsection{Conclusion}

We have presented RIDGE, an autonomous validation framework for LLM-generated option-pricing implementations, organized around a methodological blueprint, a pricing implementation, and a quantitative validation operator supported by a repository of accumulated validation expertise.

For the Heston, Bates, and rough Heston models, RIDGE validates whether generated implementations realize documented pricing methodologies and progressively accumulates adaptive numerical techniques that can be reused in subsequent studies. For the Heston--Hull--White and G-SVJD models, the validation process goes further: it identifies structural limitations of the generated Monte Carlo methodologies and guides the transition to deterministic conditional characteristic-function PDE formulations, which are subsequently refined through the same validation process.

The ablation experiments identify the validation operator as the central component of the framework. Without it, an implementation may appear accurate on the configuration against which it is refined while exhibiting substantial pricing errors and arbitrage violations elsewhere. The external-generator experiments show, conversely, that RIDGE can validate and refine implementations produced by different language models. The resulting pricing methodology need not coincide with that obtained in the original experiments; nevertheless, the same measurement and diagnosis process can guide independently generated implementations to a validated state.
Taken together, the experiments suggest that autonomous validation has a broader role than detecting implementation defects. It can reveal insufficient numerical adaptivity, recognize when a pricing methodology is unsuitable, retain transferable numerical expertise across model classes, and guide the development of alternative numerical techniques. As large language models become increasingly capable of generating scientific software from mathematical specifications, systematic validation frameworks of this kind may provide an important means of assessing the reliability of LLM-generated computational methods.
The stochastic-process interpretation of Section~\ref{ssec:stochastic} gives sufficient conditions for almost-sure convergence of the validation energy to the acceptance region. Whether these conditions hold for specific classes of language-model generators, how they translate into finite-time guarantees, and how stable the validation process remains across independently initialized runs are open questions.
	\section*{Acknowledgements}
Xue Cheng was supported by the Natural Science Foundation of China under Grant No.~11971040. Liexin Cheng gratefully acknowledges financial support from the China Scholarship Council. We would also like to thank Prof. Lech Grzelak for valuable suggestions.
	
	\bibliographystyle{unsrt}
	\bibliography{refs}
	
	\appendix
	\makeatletter
	\renewcommand\@seccntformat[1]{Appendix~\csname the#1\endcsname.\quad}
	\makeatother
	\section{The affine-proxy benchmark}\label{app:affine}
	
	\revseven{When Block~III as described in Section  \S\ref{sec:iteration} has \reveight{neither published reference prices nor a degeneration limit}, it falls back on a local affine proxy. Starting from the general variance dynamics $\mathrm{d}v_t = \mu_v(v_t)\,\mathrm{d}t + \sigma_v(v_t)\,\mathrm{d}Z_t$, the drift $\mu_v(v)$ and the squared diffusion coefficient $\sigma_v(v)^2$ are linearized around the initial variance level $v_0$,
		$
			\mu_v(v) \approx \alpha_0 + \alpha_1 v,\qquad \sigma_v(v)^2 \approx \beta_0 + \beta_1 v,
		$
		with affine coefficients from first-order Taylor expansions,
		$
			\alpha_0 = \mu_v(v_0) - \mu_v'(v_0)v_0,\quad \alpha_1 = \mu_v'(v_0),\qquad
			\beta_0 = \sigma_v(v_0)^2 - (\sigma_v^2)'(v_0)v_0,\quad \beta_1 = (\sigma_v^2)'(v_0).
		$
		The resulting proxy has affine variance drift and affine squared volatility-of-variance dynamics while inheriting the jump and correlation structure of the original model, so it admits a characteristic-function representation through the standard Riccati-equation framework of affine stochastic volatility theory~\cite{duffie2000}, and European option prices follow by the COS method. Because the linearization is local around $v_0$, the proxy is an accurate benchmark at short maturities, where the at-the-money implied volatility and skew of the implementation converge to those of the proxy.}
	
	\section{HHW: development and pricing algorithm}\label{app:hhw}
	
\revten{This appendix derives the characteristic function of the log-return under the HHW model and records the final iteration pricing algorithm.}
The HHW dynamics have been presented in Section~\ref{sec:meth_hhw}.

\rev{\begin{proposition}[HHW CCF-PDE]\label{prop:hhw}
Let $Q^T$ be the $T$-forward measure and $F_t=S_t/P(t,T)$ the forward price, with $F_T=S_T$. The zero-coupon bond $P(0,T)$ is given by~\eqref{eq:hhw_bond} below, and
\[
B_r(t,T)=\frac{1-e^{-a(T-t)}}{a}.
\]

Define \(A_{\mathrm{HW}}(u;T)=-\tfrac12(u^2+\mathrm{i}u)\,\sigma_2^2\int_0^T B_r(t,T)^2\,\mathrm{d}t.\) The characteristic function of $\log (F_T/F_0)$ is given by
\begin{equation}
\varphi(u;T)=\exp(-\mathrm{i}u\rho_{12}v_0/\sigma_1)\,h(0,v_0;u)\,\exp(A_{\mathrm{HW}}(u;T)),
\label{eq:hhw_cf}
\end{equation}
where \(h(t, v; u)\) is the solution to the one-dimensional backward parabolic equation
\begin{equation}
-\partial_t h = \kappa(\eta-v)\,\partial_v h + \tfrac12\sigma_1^2v\,\partial_{vv}h + c_{\mathrm{HHW}}(t,v;u)\,h,
\qquad h(T,v;u)=\exp(\mathrm{i}u\rho_{12}v/\sigma_1),
\label{eq:hhw_pde_prop}
\end{equation}
with 
\begin{equation}
c_{\mathrm{HHW}}(t,v;u)=-\tfrac12u^2(1-\rho_{12}^2)v
-\tfrac12\mathrm{i}u\,v
-\mathrm{i}u\,\frac{\rho_{12}\kappa(\eta-v)}{\sigma_1}
-\mathrm{i}u\,\rho_{13}\sigma_2 B_r(t,T)\sqrt{v}
-u^2\rho_{13}\sigma_2 B_r(t,T)\sqrt{v}.
\label{eq:hhw_source}
\end{equation}

\end{proposition}}

\begin{proof}
\revten{
We work under the $T$-forward measure $Q^T$, with the zero-coupon bond $P(t,T)$ as numeraire. The forward
\[
F_t=\frac{S_t}{P(t,T)}
\]
is a martingale with $F_T=S_T$. Because the Hull--White short rate makes $\int_t^T r_s\,\mathrm{d}s$ Gaussian, the bond is
\begin{equation*}
P(t,T)=\exp\!\bigl(A_{\mathrm{HW}}^P(t,T)-B_r(t,T)r_t\bigr),
\label{eq:hhw_bond}
\end{equation*}
with
\begin{equation*}
B_r(t,T)=\frac{1-e^{-a(T-t)}}{a},\qquad
A_{\mathrm{HW}}^P(t,T)=\Bigl(b-\frac{\sigma_2^2}{2a^2}\Bigr)(B_r+t-T)-\frac{\sigma_2^2B_r^2}{4a}.
\end{equation*}

The forward dynamics under $Q^T$ are obtained by changing numeraire. The Girsanov kernel $-\sigma_2B_r\,\mathrm{d}W_3$ shifts
\[
\mathrm{d}W_1^T(t)=\mathrm{d}W_1(t)+\rho_{13}\sigma_2B_r(t,T)\,\mathrm{d}t,\qquad
\mathrm{d}W_3^T(t)=\mathrm{d}W_3(t)+\sigma_2B_r(t,T)\,\mathrm{d}t,
\]
so that
\begin{equation*}
\frac{\mathrm{d}F_t}{F_t}=\sqrt{v_t}\,\mathrm{d}W_1^T(t)+\sigma_2B_r(t,T)\,\mathrm{d}W_3^T(t),
\end{equation*}
with $W_1^T,W_3^T$ standard $Q^T$-Brownian motions and $\mathrm{d}[W_1^T,W_3^T]_t=\rho_{13}\,\mathrm{d}t$.

Conditioning on the variance path $\{v_t\}_{0\le t\le T}$ makes $X\equiv\log(F_T/F_0)$ Gaussian. Integrating via It\^o's formula,
\begin{equation}
\begin{aligned}
X &=\int_0^T\frac{\mathrm{d}F_t}{F_t}-\frac12\int_0^T\Bigl(\frac{\mathrm{d}F_t}{F_t}\Bigr)^{\!2} \\[4pt]
  &=\int_0^T\sqrt{v_t}\,\mathrm{d}W_1^T(t)+\sigma_2\int_0^T B_r(t,T)\,\mathrm{d}W_3^T(t)
     -\frac12\int_0^T v_t\,\mathrm{d}t-\frac12\sigma_2^2 K_B-\rho_{13}\sigma_2\int_0^T B_r(t,T)\sqrt{v_t}\,\mathrm{d}t,
\end{aligned}
\end{equation}
where $K_B=\int_0^T B_r(t,T)^2\,\mathrm{d}t$. Decompose $W_1^T(t)=\rho_{12}W_2(t)+\sqrt{1-\rho_{12}^2}\,W^\perp(t)$ with $W^\perp$ independent of the variance driver $W_2$,
\begin{equation}
\begin{aligned}
X &= \rho_{12}\int_0^T\sqrt{v_t}\,\mathrm{d}W_2(t)+\sqrt{1-\rho_{12}^2}\int_0^T\sqrt{v_t}\,\mathrm{d}W^\perp(t) \\
  &\quad +\sigma_2\int_0^T B_r(t,T)\,\mathrm{d}W_3^T(t)
     -\frac12\int_0^T v_t\,\mathrm{d}t-\frac12\sigma_2^2 K_B-\rho_{13}\sigma_2\int_0^T B_r(t,T)\sqrt{v_t}\,\mathrm{d}t.
\end{aligned}
\end{equation}
Conditional on $\{v_t\}$, the $W^\perp$ and $W_3^T$ integrals are jointly Gaussian. Since $W_1^T=\rho_{12}W_2+\sqrt{1-\rho_{12}^2}\,W^\perp$ and $\mathrm{d}[W_1^T,W_3^T]_t=\rho_{13}\,\mathrm{d}t$, we have $\mathrm{d}[W^\perp,W_3^T]_t=\rho_{13}/\sqrt{1-\rho_{12}^2}\,\mathrm{d}t$, so their covariance is $\rho_{13}\sigma_2\int_0^T B_r(t,T)\sqrt{v_t}\,\mathrm{d}t$. The obstacle is the $\rho_{12}$-coupled term $\int_0^T\sqrt{v_t}\,\mathrm{d}W_2(t)$, driven by the same Brownian motion as the variance.

The linear gauge $g(v)=v/\sigma_1$ removes this term: $g'(v)\sigma_1\sqrt{v}=\sqrt{v}$ and $g''\equiv0$, so It\^o's formula gives
\begin{equation*}
\int_0^T\sqrt{v_s}\,\mathrm{d}W_2(s)=\frac{v_T-v_0}{\sigma_1}-\frac{\kappa}{\sigma_1}\int_0^T(\eta-v_s)\,\mathrm{d}s.
\end{equation*}
Substituting this identity yields the conditional law
\begin{equation}
X\mid\{v_t\}\;\sim\;\mathcal{N}\!\Bigl(\rho_{12}\frac{v_T-v_0}{\sigma_1}+\int_0^T f(v_s)\,\mathrm{d}s-\frac12\sigma_2^2 K_B,\;
(1-\rho_{12}^2)\int_0^T v_s\,\mathrm{d}s+\sigma_2^2 K_B+2\rho_{13}\sigma_2\int_0^T B_r\sqrt{v_s}\,\mathrm{d}s\Bigr),
\label{eq:hhw_condlaw}
\end{equation}
with running coefficient
\[
f(v)=-\frac{\rho_{12}\kappa(\eta-v)}{\sigma_1}-\frac12 v-\rho_{13}\sigma_2B_r\sqrt{v}.
\]

The gauge identity contributes the term $\rho_{12}(v_T-v_0)/\sigma_1$ to the mean of $X$. Taking the conditional characteristic function of~\eqref{eq:hhw_condlaw},
\[
\begin{aligned}
\mathbb{E}\bigl[\exp(\mathrm{i}uX)\mid\{v_t\}\bigr]
&= \exp\!\Bigl(\mathrm{i}u\rho_{12}\frac{v_T-v_0}{\sigma_1}\Bigr)
   \exp\!\Bigl(-\tfrac12(u^2+\mathrm{i}u)\sigma_2^2K_B\Bigr) \\
&\quad\times\exp\!\Bigl(\int_0^T\bigl[\mathrm{i}u f(v_s)-\tfrac12 u^2(1-\rho_{12}^2)v_s
   -u^2\rho_{13}\sigma_2B_r\sqrt{v_s}\bigr]\mathrm{d}s\Bigr).
\end{aligned}
\]
The second factor is $\exp(A_{\mathrm{HW}}(u;T))$, collecting the deterministic Hull--White contribution. Define
\[
h(t,v;u)=\mathbb{E}^{Q^T}\!\Bigl[\exp\!\Bigl(\mathrm{i}u{\frac{v_T}{\rho_{12}\sigma_1}}
+\int_t^T c_{\mathrm{HHW}}(s,v_s;u)\,\mathrm{d}s\Bigr)\Bigm| v_t=v\Bigr],
\]
with $c_{\mathrm{HHW}}$ as in~\eqref{eq:hhw_source}. By the Feynman--Kac theorem, $h$ satisfies the backward equation~\eqref{eq:hhw_pde_prop} with terminal condition $h(T,v;u)=\exp(\mathrm{i}u\rho_{12}v/\sigma_1)$. Multiplying $h(0,v_0;u)$ by $\exp(-\mathrm{i}u\rho_{12}v_0/\sigma_1)$ and $\exp(A_{\mathrm{HW}}(u;T))$ recovers $\varphi(u;T)$ of~\eqref{eq:hhw_cf}.}
\end{proof}
	
	\begin{algorithm}[H]
		\caption{\revten{HHW European option pricing via the CCF-PDE of Proposition~\ref{prop:hhw}.}}
		\begin{algorithmic}[1]
			\Require model parameters $(\kappa,\eta,\sigma_1,\rho_{12},\rho_{13},\sigma_2,a,b)$; market $(S_0,v_0,r_0)$; maturity $T$; strikes $\{K_j\}$
			\Ensure call and put prices $\{V_{c,j}\},\{V_{p,j}\}$
			
			\Statex \emph{Notation.} $f=2\kappa\eta/\sigma_1^2$ is the Feller ratio; $N$ is the number of Fourier modes with largest node $u_{\max}$; $M$ is the number of variance cells; $N_t$ is the number of time steps; $K_B=\int_0^T B_r(t,T)^2\,\mathrm{d}t$, with $B_r$ as in~\eqref{eq:hhw_bond}.
			
			\Statex \textbf{Phase 1 -- COS domain, Fourier modes, and variance grid.}
			\State compute $P(0,T)$ from~\eqref{eq:hhw_bond} and set $F_0 \gets S_0/P(0,T)$
			\State estimate the first two cumulants $\bar c_1$ and $\bar c_2$ of $\log(F_T/F_0)$ from central differences of $\log\varphi$ near $u=0$; \quad $s_x \gets \sqrt{\bar c_2}$
			\State $L_{\mathrm{eff}} \gets
			\begin{cases}
				10, & f\ge 1,\\
				\min\{10(2.5-f),\,16\}, & f<1
			\end{cases}$
			\State $\mathrm{half} \gets \max\{\,L_{\mathrm{eff}}s_x,\ 0.7,\ \max_j|\log(F_0/K_j)|+0.1\,\}$; \quad $[a_c,b_c] \gets [\,\bar c_1-\mathrm{half},\ \bar c_1+\mathrm{half}\,]$
			\State $N \gets \mathrm{clip}(24\,\mathrm{half}/s_x,\,128,\,1024)$, raised to $\mathrm{clip}(N(2.5-f),\,N,\,4096)$ when $f<1$; \quad $u_k \gets k\pi/(b_c-a_c)$, $k=0,\dots,N-1$
			\State $v_{\max} \gets \max\{\,\mathbb{E}[v_T]+6\sqrt{\mathrm{Var}[v_T]},\ 2\max(v_0,\eta)\,\}$ from the closed-form CIR moments
			\State construct a uniform variance grid $0=\xi_0<\dots<\xi_M=v_{\max}$, refined until $\Delta\xi \le (\pi/4)\,\sigma_1/(|\rho_{12}|\,u_{\max})$
			
			\Statex \textbf{Phase 2 -- One backward PDE solve shared across all Fourier modes and strikes.}
			\State set the terminal values $\Psi_k(\xi_m) \gets \exp(\mathrm{i}u_k\rho_{12}\xi_m/\sigma_1)$ for all $k,m$, as in~\eqref{eq:hhw_pde_prop}
			\State assemble the finite-difference discretization of $\mathcal{L}=\kappa(\eta-v)\partial_v+\tfrac12\sigma_1^2v\,\partial_{vv}$ by central differences, with a one-sided transport row at $v=0$ and the Neumann ghost condition $\partial_v\Psi=\mathrm{i}u\rho_{12}\Psi/\sigma_1$ at $v_{\max}$
			\State recenter $\operatorname{Im}c_{\mathrm{HHW}}$ of~\eqref{eq:hhw_source} along the mean path $m(t)=\eta+(v_0-\eta)e^{-\kappa t}$, absorbing its time integral into the prefactor while leaving $\operatorname{Re}c_{\mathrm{HHW}}$ unchanged
			\For{$n=1$ \textbf{to} $N_t$}
			\State advance all modes $\{\Psi_k\}_{k=0}^{N-1}$ one Crank--Nicolson step in $\mathcal{L}+\operatorname{diag}c_{\mathrm{HHW}}(\cdot\,;u_k)$ by a single Thomas sweep batched over $k$
			\EndFor
			\State recover the characteristic-function values $\varphi(u_k;T)$ from~\eqref{eq:hhw_cf}, with $\Psi_k(v_0)$ interpolated and the recentering integral restored
			
			\Statex \textbf{Phase 3 -- COS valuation with out-of-the-money pricing and put--call parity.}
			\State $V_{c,j}^{\mathrm{otm}} \gets$ COS call series on $[0,b_c]$, \quad $V_{p,j}^{\mathrm{otm}} \gets$ COS put series on $[a_c,0]$, both from $\{\varphi(u_k;T)\}$ and discounted by $P(0,T)$
			\State $V_{c,j} \gets
			\begin{cases}
				V_{c,j}^{\mathrm{otm}}, & K_j \ge F_0,\\
				V_{p,j}^{\mathrm{otm}}+P(0,T)(F_0-K_j), & K_j < F_0
			\end{cases}$
			\State $V_{p,j} \gets V_{c,j}-P(0,T)(F_0-K_j)$
			\State \textbf{return} $\{V_{c,j}\},\{V_{p,j}\}$
		\end{algorithmic}
	\end{algorithm}
	
	\section{G-SVJD: development and pricing algorithm}\label{app:gsvjd}
	
\revten{This appendix presents the derivation of the G-SVJD conditional characteristic-function PDE of Proposition~\ref{prop:gsvjd} and records the final iteration pricing algorithm.}

The G-SVJD dynamics are those of Section~\ref{sec:meth_gsvjd}. Define the following quantities:
\begin{enumerate}
  \item the log-normal jump characteristic function $\psi_J(u)=\exp(\mathrm{i}u\mu_y-\tfrac12u^2\sigma_y^2)$;
  \item the gauge function $g(v)=v^{3/2-p}/(\sigma_v(3/2-p))$, $p\neq\tfrac32$, satisfying $g'(v)\sigma_vv^p=\sqrt{v}$;
  \item the drift correction
  \begin{equation}
  F(v)=(r-q-\lambda m_J)-\frac{v}{2}-\frac{\rho\kappa(\bar\nu-v)v^{1/2-p}}{\sigma_v}
  -\frac{\rho}{2}\Bigl(\tfrac12-p\Bigr)\sigma_vv^{p-1/2};
  \label{eq:gsvjd_drift}
  \end{equation}
  \item the path functionals
  \[
  A_{T}=\rho[g(v_T)-g(v_0)]+\int_0^T F(v_s)\,\mathrm{d}s,\qquad
  B_T=\int_0^T v_s\,\mathrm{d}s.
  \]
\end{enumerate}

\rev{\begin{proposition}[G-SVJD CCF-PDE]\label{prop:gsvjd}
{The characteristic function of $\log(S_T/S_0)$ is given by
\begin{equation}
\varphi(u;T)=\exp(-\mathrm{i}u\rho g(v_0))\,h(0,v_0;u)\,\exp\!\bigl(\lambda T(\psi_J(u)-1)\bigr),
\label{eq:gsvjd_cf}
\end{equation}
where $h(t, v; u)$ solves the one-dimensional backward parabolic equation
\begin{equation}
-\partial_t h = \kappa(\bar\nu-v)\,\partial_v h + \tfrac12\sigma_v^2v^{2p}\,\partial_{vv}h + c(v;u)\,h,
\qquad c(v;u)=\mathrm{i}uF(v)-\tfrac12u^2(1-\rho^2)v,
\label{eq:gsvjd_pde}
\end{equation}
with terminal condition $h(T,v;u)=\exp(\mathrm{i}u\rho g(v))$ and $\mathrm{Re}\,c\le0$.}
\end{proposition}}
	
\begin{proof}
\revten{
Let $X\equiv\log(S_T/S_0)$ be the log-return. The jump contribution factors as $\exp(\lambda T(\psi_J(u)-1))$, so it suffices to treat the diffusive part of $X$. Write $W=\rho Z+\sqrt{1-\rho^2}\,W^\perp$ with $W^\perp$ independent of $Z$. Integrating the log-return SDE,
\[
X=(r-q-\lambda m_J)T-\frac12\int_0^T v_s\,\mathrm{d}s
+\rho\int_0^T\sqrt{v_s}\,\mathrm{d}Z_s+\sqrt{1-\rho^2}\int_0^T\sqrt{v_s}\,\mathrm{d}W_s^\perp .
\]
The $W^\perp$ integral is, conditional on $\{v_s\}_{0\le s\le T}$, centered Gaussian with variance $(1-\rho^2)\int_0^T v_s\,\mathrm{d}s$. The obstacle is $\int_0^T\sqrt{v_s}\,\mathrm{d}Z_s$, driven by the same Brownian motion as the variance.

Applying It\^o's formula to $g(v_t)$ along $\mathrm{d}v_t=\kappa(\bar\nu-v_t)\,\mathrm{d}t+\sigma_vv_t^p\,\mathrm{d}Z_t$,
\[
\mathrm{d}g(v_t)=\bigl[g'(v_t)\kappa(\bar\nu-v_t)+\tfrac12g''(v_t)\sigma_v^2v_t^{2p}\bigr]\mathrm{d}t+g'(v_t)\sigma_vv_t^p\,\mathrm{d}Z_t.
\]
Since $g'(v_t)\sigma_vv_t^p=\sqrt{v_t}$, the stochastic term is $\sqrt{v_t}\,\mathrm{d}Z_t$. Integrating,
\[
\int_0^T\sqrt{v_s}\,\mathrm{d}Z_s=g(v_T)-g(v_0)-\int_0^T\bigl[g'(v_s)\kappa(\bar\nu-v_s)+\tfrac12g''(v_s)\sigma_v^2v_s^{2p}\bigr]\mathrm{d}s.
\]

Substituting into $X$ and collecting the deterministic drift into $F(v)$ yields
\[
X=\rho[g(v_T)-g(v_0)]+\int_0^T F(v_s)\,\mathrm{d}s+\sqrt{1-\rho^2}\int_0^T\sqrt{v_s}\,\mathrm{d}W_s^\perp
=A_T+\sqrt{1-\rho^2}\int_0^T\sqrt{v_s}\,\mathrm{d}W_s^\perp .
\]
Hence, conditional on $\{v_s\}$,
$
X\mid\{v_s\}\;\sim\;\mathcal{N}\!\bigl(A_T,\;(1-\rho^2)B_T\bigr).
$

The conditional characteristic function is thus given by
\[
\mathbb{E}\bigl[\exp(\mathrm{i}uX)\mid\{v_s\}\bigr]=\exp\!\bigl(\mathrm{i}uA_T-\tfrac12u^2(1-\rho^2)B_T\bigr).
\]
Taking the outer expectation over the variance path defines $h$ as
{$$h(t, v; u):= \mathbb{E}\!\Bigl[\exp{\!\Bigl(\mathrm{i}u\rho g(v_T) + \int_t^T c(v_s;u)\Bigr)\,\mathrm{d}s}\mid v_t = v\Bigr].$$}
By the Feynman--Kac theorem, $h$ satisfies the backward equation~\eqref{eq:gsvjd_pde} with terminal condition $h(T,v;u)=\exp(\mathrm{i}u\rho g(v))$ and $\mathrm{Re}\,c=-\tfrac12u^2(1-\rho^2)v\le0$. {Multiplying by the prefactor $\exp(-\mathrm{i}u\rho g(v_0))$ left from} $A_T$ and restoring the jump factor gives $\varphi(u;T)$ of~\eqref{eq:gsvjd_cf}.}
\end{proof}
	
	Two further {numerical techniques} are developed for the variance process. First, the diffusion coefficient $\sigma_v^2v^{2p}$ in~\eqref{eq:gsvjd_pde} degenerates as $v\downarrow0$. The PDE is solved in the Lamperti coordinate
		\[
		Y=\frac{v^{1-p}}{\sigma_v(1-p)},
		\]
		which transforms the state-dependent diffusion coefficient into a constant one and concentrates grid resolution near the origin, where the solution varies most rapidly. Second, the drift correction $F$ of~\eqref{eq:gsvjd_drift} contains the factor $v^{1/2-p}$, which becomes singular at $v=0$ when $p>\tfrac12$. This singularity is handled by replacing $F$ at the first interior node by its average over the adjacent cell, while leaving the interior discretization unchanged.
	\begin{algorithm}[H]
		\caption{\revten{G-SVJD European option pricing via the CCF-PDE of Proposition~\ref{prop:gsvjd}.}}
		\begin{algorithmic}[1]
			\Require model parameters $(\kappa,\bar\nu,\sigma_v,p,\rho,\lambda,\mu_y,\sigma_y)$; market $(S_0,v_0,r,q)$; maturity $T$; strikes $\{K_j\}$
			\Ensure call and put prices $\{V_{c,j}\},\{V_{p,j}\}$
			
			\Statex \emph{Notation.} $L$ is the COS truncation multiple; $N$ is the number of Fourier modes; $M$ is the number of variance cells; $N^{\mathrm{yr}}$ is the number of time steps per unit maturity and $N_t$ the total number of time steps; $m_J=\exp(\mu_y+\sigma_y^2/2)-1$ is the mean jump size; $g$, $F$, and $\psi_J$ are as defined above.
			
			\Statex \textbf{Phase 1 -- COS domain and Fourier modes.}
			\State $\bar c_1 \gets (r-q-\lambda m_J)T+\lambda T\mu_y-\tfrac12\bar\nu T$; \quad $\bar c_2 \gets \bar\nu T+\lambda T(\mu_y^2+\sigma_y^2)$
			\State $w \gets L\sqrt{\bar c_2}+0.15\sqrt{T}$; \quad $[a,b]\gets[\,\bar c_1-w,\ \bar c_1+w\,]$, enlarged if necessary so that every $\log(K_j/S_0)$ lies inside the interval with margin
			\State $u_k \gets k\pi/(b-a)$ for $k=0,\dots,N-1$
			
			\Statex \textbf{Phase 2 -- One backward PDE solve shared across all Fourier modes and strikes.}
			\State construct a Lamperti grid, uniform in $Y=v^{1-p}/[\sigma_v(1-p)]$ and mapped to variance nodes $0=\xi_0<\dots<\xi_M=v_{\max}$, with $v_{\max}$ chosen from the variance-process moments; this clusters nodes near $v=0$, where $\sigma_vv^p$ degenerates
			\State set $F$ from~\eqref{eq:gsvjd_drift}, replacing its value at the first interior node by the cell average over $[0,\xi_{1/2}]$ to regularize the $v^{1/2-p}$ singularity when $p>\tfrac12$
			\State set terminal data $h_k(\xi_m)\gets \exp(\mathrm{i}u_k\rho\,g(\xi_m))$ for all $k,m$, as in~\eqref{eq:gsvjd_pde}
			\State assemble the finite-difference discretization of $\mathcal{L}=\kappa(\bar\nu-v)\partial_v+\tfrac12\sigma_v^2v^{2p}\partial_{vv}$ by non-uniform central differences with constant-preserving boundaries, and add the source $c(\cdot\,;u)$ of~\eqref{eq:gsvjd_pde}
			\State $N_t\gets\max\{\lceil N^{\mathrm{yr}}T\rceil,\,24\}$
			\For{$n=1$ \textbf{to} $N_t$}
			\State advance all modes $\{h_k\}_{k=0}^{N-1}$ one step in $\mathcal{L}+\operatorname{diag}c(\cdot\,;u_k)$ by a Thomas sweep batched over $k$, using implicit Euler (Rannacher) for the first two steps and Crank--Nicolson thereafter
			\EndFor
			\State recover the characteristic-function values $\varphi(u_k;T)$ from~\eqref{eq:gsvjd_cf}, with $h_k(v_0)$ interpolated
			
			\Statex \textbf{Phase 3 -- COS valuation with put--call parity.}
			\State for each strike $K_j$, compute $V_{c,j}$ from the COS call series, integrating the payoff from $\log(K_j/S_0)$ to $b$ using $\{\varphi(u_k;T)\}$
			\State $V_{p,j}\gets V_{c,j}-S_0\exp(-qT)+K_j\exp(-rT)$
			\State \textbf{return} $\{V_{c,j}\},\{V_{p,j}\}$
		\end{algorithmic}
	\end{algorithm}
	
	
	\FloatBarrier
	\section{Stress-test design and corner cases}\label{app:stress_design}
	
	This appendix complements \S\ref{sec:iteration} by detailing the development of the stress-test parameter sets and the placement of the corner cases.
	
	Let $\boldsymbol{p}=(p_1,\ldots,p_n)$ denote the free parameters exposed by the implementation interface, and let $\boldsymbol{p}^{(0)}$ denote the blueprint default. The two observables are the at-the-money implied volatility $\sigma_{\mathrm{IV}}(T)$ and the at-the-money skew
	\[
	\mathrm{Sk}(T)
	=
	\left.
	\frac{\partial \sigma_{\mathrm{IV}}}{\partial \log K}
	\right|_{K=F},
	\]
	each evaluated at the reference maturities $T_1=0.1$ and $T_2=1.0$. Writing $\mathcal{O}^{(\ell)}(T_m;\boldsymbol{p})$ for the $\ell$-th observable, the sensitivity
	\[
	\frac{\partial \mathcal{O}^{(\ell)}}{\partial p_i}
	\]
	at $\boldsymbol{p}^{(0)}$ is computed numerically. Four target levels are prescribed,
	$
	\bar{\mathcal{O}}^{(1)}\in\{0.50,\,0.10\},
	\;
	\bar{\mathcal{O}}^{(2)}\in\{-2.0,\,0.0\},
	$
	for $\sigma_{\mathrm{IV}}$ and $\mathrm{Sk}$, respectively. For each parameter, maturity, and target, the candidate displacement is
	\[
	\Delta p_i^{(\ell,m,\bar{\mathcal{O}})}
	=
	\frac{
		\bar{\mathcal{O}}^{(\ell)}
		-
		\mathcal{O}^{(\ell)}(T_m;\boldsymbol{p}^{(0)})
	}{
		\left.
		\partial \mathcal{O}^{(\ell)}(T_m;\boldsymbol{p})/\partial p_i
		\right|_{\boldsymbol{p}^{(0)}}
	}.
	\]
	
	Across all combinations, this procedure determines, for each parameter $p_i$, a reachable extreme in each direction, namely the smallest positive and the smallest negative displacement, both clamped to the admissible parameter range prescribed by the blueprint. The computed displacement determines the direction and order of magnitude of the stress, while the value ultimately used is rounded to a representative figure within that range so that the stress parameters appear as round numbers.
	
	Each iteration stresses at most five parameters, selected as the most representative of the model structure, retaining both extremes of each parameter and  generating at most ten stress configurations. Up to two documented abnormal regimes, such as violations of the Feller condition~\cite{feller1951} in Heston-type models, are then appended. Including the base configuration, an iteration evaluates the base case, up to ten stress configurations, and up to two abnormal regimes.
	
	The corner cases consist of four isolated points placed outside the standard strike--maturity window, each designed to isolate a single extreme rather than to form a maturity-by-strike corner. They comprise the two extreme maturities,
$T=0.001
	\; \text{and}\;
	T=5.0,$
	evaluated at the money, where the price remains well defined while the time discretization and cosine truncation are under the greatest strain, together with the deep-strike pair
	$
	(T,\log(K/F))
	=
	(0.5,\pm0.4),
	$
	evaluated at a representative maturity and lying outside the standard moneyness grid.

\end{document}